\documentclass[aps,twocolumn,prx,preprintnumbers, nofootinbib,floatfix,10pt]{revtex4-2}
\usepackage{xpatch}
\makeatletter
\patchcmd{\@ssect@ltx}
    {\addcontentsline{toc}{#1}{\protect\numberline{}#8}}
    {}
    {}
    {}
\makeatother
\pdfoutput=1
\usepackage{color}
\usepackage{enumitem}
\usepackage[dvipsnames]{xcolor}
\usepackage{hyperref}
\definecolor{oxfordblue}{RGB}{0,33,71}
\definecolor{amaranth}{rgb}{0.9, 0.17, 0.31}
\definecolor{forestForestGreen(web)}{rgb}{0.13, 0.55, 0.13}
\definecolor{blue(munsell)}{HTML}{005567}
\definecolor{bblue}{rgb}{0.0, 0.58, 0.71}
\hypersetup{
pdfstartview={FitH}, 
hypertexnames=false, 
pdftitle={}, 
colorlinks=true, 
linkcolor=blue(munsell), 
citecolor=oxfordblue, 
filecolor=oxfordblue, 
urlcolor=oxfordblue
}
\usepackage{pgfplots}
\pgfplotsset{compat=1.18}
\usepackage{amssymb}
\usepackage{amsfonts}
\usepackage{graphicx}
\usepackage{epstopdf}
\usepackage{dcolumn}
\usepackage{amssymb}
\usepackage{amsfonts}
\usepackage{graphicx}
\usepackage{epstopdf}
\usepackage{dcolumn}
\usepackage{amsmath}
\usepackage{mathtools}
\usepackage{latexsym,bm}
\usepackage{amsthm}
\usepackage{latexsym,bm}
\usepackage{amsthm}
\usepackage{slashed}
\usepackage{color}
\usepackage{url}
\usepackage{longtable}
\usepackage{rotating}
\usepackage[normalem]{ulem}
\usepackage{faktor}
\usepackage{tikz-cd}
\usepackage{tensor}
\usepackage{array}
\usepackage{tabularx}
\usepackage{makecell}
\usepackage{changepage}

\usepackage{pgfplots}

\makeatletter
\renewcommand{\p@subsection}{}
\renewcommand{\p@subsubsection}{}
\g@addto@macro\appendix{\renewcommand{\p@subsection}{}}
\makeatother

\numberwithin{equation}{section}

\newcommand{\de}{\partial}

\newcommand{\paragraf}[1]{\par\medskip\noindent{\bf #1}}
\numberwithin{equation}{section}

\long\def\remove#1{{\color [rgb]{.6,.2,.1}#1}}

\usepackage{tikz}
\usepackage{diagbox}
\usetikzlibrary{positioning}
\usetikzlibrary{calc}
\usetikzlibrary{decorations.pathreplacing,calligraphy}

\usepackage{xstring}
\usetikzlibrary{decorations.pathmorphing} 
\usetikzlibrary{decorations.markings} 
\usetikzlibrary{arrows} 
\usetikzlibrary{shapes} 
\usetikzlibrary{matrix} 
\usetikzlibrary{positioning} 
\usepackage[english]{babel} 
\usepackage[autostyle]{csquotes}
\usepackage{pifont}
\usetikzlibrary{shapes.multipart}

\tikzset{->-/.style={decoration={
  markings,
  mark=at position .5 with {\arrow{>}}},postaction={decorate}}}

\newcommand{\bea}{\begin{eqnarray}}
\newcommand{\eea}{\end{eqnarray}}
\newcommand{\be}{\begin{equation}}
\newcommand{\ee}{\end{equation}}
\newcommand{\ba}{\begin{aligned}}
\newcommand{\ea}{\end{aligned}}
\newcommand{\bit}{\begin{itemize}}
\newcommand{\eit}{\end{itemize}}
\newcommand{\ben}{\begin{enumerate}}
\newcommand{\een}{\end{enumerate}}
\newcommand{\nn}{\nonumber}
\newcommand{\id}{\text{id}}

\renewcommand{\ol}{\overline}

\newcommand{\Tr}{\text{Tr}}

\newcommand{\bbI}{\mathbb{I}}

\newcommand{\wt}{\widetilde}

\newcommand{\ot}{\otimes}

\newcommand{\Z}{{\mathbb Z}}

\newcommand{\bC}{{\mathbb C}}

\newcommand{\cE}{\mathcal{E}}
\newcommand{\cF}{\mathcal{F}}

\renewcommand{\dim}{\text{dim}}

\newcommand{\Irr}{\text{Irr}}

\usepackage{comment}

\def\remove#1{{\color [rgb]{1,0,0}\sout{#1}}}
\global\long\def\remove#1{}

\usepackage{physics}

\makeatother

\newcommand{\III}{\text{III}}

\usepackage{makecell}

\newtheorem{theorem}{Theorem}
\newtheorem{proposition}[theorem]{Proposition}
\newtheorem{lemma}[theorem]{Lemma}

\theoremstyle{definition}

\makeatletter
\newcommand\xlabel[2][]{\phantomsection\def\@currentlabelname{#1}\label{#2}}
\makeatother

\makeatletter
\renewcommand\subsubsection{\@startsection{subsubsection}{3}{\z@}%
                                     {-3.25ex\@plus -1ex \@minus -.2ex}%
                                     {2.5ex \@plus .2ex}%
                                     {\centering\itshape\bfseries\small}}
\makeatother

\begin{document}
\title{Exact Recovery for Non-Abelian Surface Codes} 

\author{Alison Warman$\,^1$}
\author{Nathanan Tantivasadakarn$\,^2$}
\author{Sakura Sch\"afer-Nameki$\,^1$}

\affiliation{$^1$ Mathematical Institute, University of Oxford, Woodstock Road, Oxford, OX2 6GG, United Kingdom}
\affiliation{$^2$ C. N. Yang Institute for Theoretical Physics, Stony Brook University, Stony Brook, NY 11794, USA}

\begin{abstract}
\noindent
We study exact recovery for non-Abelian topological surface codes based on the quantum double $D(G)$ of any finite group $G$. 
We determine an orthogonal error basis, that is comprised of group multiplication and irreducible representation operators. 
Whenever the errors are supported on closed loops on the lattice or dual lattice, there is a redundancy in the basis, due to stabilizers. To remedy this, we introduce a gauge-fixing that results in a complete and orthogonal error basis on the code space.  
Assuming as input predetermined, neutral correctable error clusters, we construct charge and flux transfer circuits that move the errors onto ancillas that then get projected out. 
This is an exact and deterministic recovery protocol, applicable to surface codes for any finite, in particular non-Abelian, group $G$. 
\end{abstract}

\maketitle
\tableofcontents

\section{Introduction}

Nearly 30 years ago, Kitaev introduced the 2D surface code for finite groups $G$ based on the quantum double $D(G)$ \cite{Kitaev:1997wr}. This is a framework that incorporates both abelian and non-abelian  topological codes. However, as surface codes (as opposed to using braiding to do quantum computing) the main focus in the literature has thus far been on abelian groups.  
Unfortunately, for 2D Pauli stabilizer codes, including the standard $G=\Z_2$ surface code, constant-depth geometrically local logicals are restricted to be  Clifford by the Bravyi-K\"onig theorem \cite{Bravyi:2012rnv}. This means that in order to achieve universal computation one requires additional resources, such as code switching to higher dimensional topological codes that admit transversal non-Clifford gates~\cite{Bombindimensionaljump16}, or non-topological protocols such as magic state distillation or cultivation~\cite{BravyiKitaev05,GidneyCultivation24}.

Recently, topological codes that host non-Abelian anyons have become more prominent. Efficient state preparation protocols for such non-Abelian states using sequential circuits or adaptive circuits have been developed \cite{Liustringnet22,Tantivasadakarn:2021vel,Verresen:2021wdv,Bravyi:2022zcw,Tantivasadakarn:2022ceu,Tantivasadakarn:2022hgp,Lyons25,Ren25}, as well as state-of-the-art quantum processors have become mature enough to experimentally realize such non-Abelian states~\cite{Iqbal:2023wvm,xu_non-abelian_2024,minev2025realizing,Lo:2026oxm}. This has motivated the use of non-abelian states as magic states \cite{Laubscher:2019rss, Davydova:2025ylx, Huang:2025cvt}, and as means to realize non-Clifford gates via generalized lattice surgery for more general groups \cite{Huang:2025ump, Manjunath:2026swm}. Perhaps more surprisingly, such non-Abelian topological codes can also admit non-Clifford (weakly) transversal gates \cite{Kobayashi:2025cfh, Warman:2025hov, Fechisin:2026ywx}. 
As was shown in~\cite{Warman:2025hov} with specific choices of group and boundary conditions, they can encode a logical qubit and realize non-Clifford gates at an arbitrary level of the Clifford hierarchy via constant-depth circuits. 

However, QEC for non-abelian codes is challenging, for various reasons.
Firstly, there is no abelian stabilizer group, since the standard $A_v^g$ and $B_p^f$ checks do not commute and do not form a group (in particular the $B_p^g$ operators are non-invertible). 
A generalization of a stabilizer group that closes and has inverses can be constructed for certain groups such as the dihedral groups  \cite{Warman:2025hov}, 
but is nevertheless non-abelian. 
Furthermore, the anyon strings of these codes compose based on a non-invertible fusion and thus, unlike standard Pauli strings, do not have inverses -- in particular the associated ribbon operators that create a definite anyon type are in general not unitary. This poses a significant challenge for sound and efficient quantum error correction (QEC).

For specific non-abelian codes, such as the twisted quantum double $D(\Z_2^3, \omega_{\III})$ of an abelian group with type III cocycle, which is equivalent as a phase of matter to the non-abelian quantum double $D(D_4)$ for the dihedral group of order $8$, so-called just-in-time decoders have been developed \cite{Davydova:2025ylx,delaFuente:2026slp}. However these rely on the specific structure of this code that uses a formulation with $\Z_2^3$ as the underlying group. For more general codes such as those in \cite{Warman:2025hov} or other finite non-abelian groups, this is insufficient.

In this paper, we consider a general finite group $G$ and develop a complete error basis for the surface code $D(G)$, and the exact correction unitary.
An important input -- which we do not determine and leave for future work -- is a syndrome extraction circuit and decoder that can provide us with neutral, correctable error clusters. Then, within each cluster, we provide a circuit that moves the errors onto ancillas that can then be decoupled.  
The main result of this paper is generally applicable for any finite group $G$. As this point is rather important to understand the relevance of the present work, we will summarize the current situation in the literature: 

\paragraf{Previous results on Non-Abelian Quantum Error Correction.}  
Here, we summarize existing results on correctability and syndrome extraction for non-Abelian codes. 
The topological property of the $D(G)$, implies correctability of local errors for every finite $G$ \cite{Cui:2019lvb} and recovery thresholds under sufficiently weak local channels have also been established \cite{PKim:2026smb}.  
Although these previous works prove that a recovery {\bf exists}, they do not construct it.\footnote{Explicit correction is  known for Abelian $G$ \cite{Cui:2024kme}.}
Syndrome extraction circuits for coarse charge and flux measurements are available for general groups \cite{Brennen:2009vgn,Manjunath:2026swm}, with detailed treatments of $D(S_3)$ in \cite{Laubscher:2019rss,Chen:2024xkw,Sajith:2025rvy}. 
Such charge and flux outcomes, however, do not resolve the full error.  Charge measurement 
is also possible using constant-depth adaptive circuits for solvable groups \cite{Bravyi:2022zcw,Lyons25,Ren25}. 
This is distinct from the complete syndrome measurements on supplied supports that we consider in the present work. Beyond groups, codes hosting Fibonacci anyons have also been considered \cite{Burton:2015bca,Schotte22}, but they however either are done abstractly without a lattice model or a assume only a single local error per time step. 

QEC strategies for non-abelian codes often require repeated rounds, even under ideal measurement assumption. It is useful to recall these approaches: anyon-level decoders fuse and remeasure iteratively \cite{Wootton:2014kgh,Wootton:2015but,Dauphinais:2016loh} 
and for the $D(D_4)= D(\Z_2^3, \omega_{\text{III}})$, errors can be corrected by Pauli strings with just-in-time decoding
\cite{Davydova:2025ylx,Kobayashi:2025cfh,Warman:2025hov,delaFuente:2026slp}.
Optimal and integer-programming decoders trace out the  internal degrees of freedom of the anyon, however, 
full recovery generally needs repeated rounds of error correction \cite{Jing:2025zuu,Jing:2026uil}. Explicit recovery channels for general surface codes are posed as future work in \cite{PKim:2026smb}. 
Neutral clusters of $D(S_3)$ can be corrected by locally ungauging to $D(\mathbb{Z}_3)$~\cite{Lyons:2026ljw}, a method expected to extend to quantum double of solvable groups, which can be prepared with constant-depth adaptive circuits \cite{Bravyi:2022zcw,Tantivasadakarn:2022hgp}.

\paragraf{Present work.} 
With this overview of the existing literature in mind, it is clear now that 
it would be important to derive an exact, deterministic, constant-round recovery for $D(G)$ for a general finite group $G$. 
A particular complication with non-abelian groups is that anyons fuse non-invertibly. 
Syndromes that measure only anyon occupation will miss their internal degrees of freedom \cite{Jing:2025zuu}. 
It is precisely this important gap that we will fill in the literature. As we have stressed before, we assume a decoder that provides us with neutral correctable error clusters, and we assume perfect syndrome measurements. Recovery then proceeds by moving errors onto plaquette and vertex ancillas while moving along paths of the dual or direct lattice respectively.
Concretely, for every finite group $G$, we determine the following error basis and recovery protocols for Kitaev quantum doubles $D(G)$:
\begin{enumerate}
    \item  
    {\bf Complete Error Basis.} In  Sec.~\ref{sec:errors} we define an orthogonal error basis and determine  
    an explicit ``gauge-fixing'' procedure to remove stabilizer redundancies. The resulting error family obeys the Knill-Laflamme Quantum Error Conditions of Thm.~\ref{thm:Knill-Laflamme} and is an orthogonal basis on the code space. 

\item {\bf Error Transfer Protocols.} In Sec.~\ref{sec:errorcorrection} we introduce so-called flux and charge transfer circuits. These move errors onto ancillas that are then projected out. E.g. any multiplication error on an edge that is the only errored edge of a plaquette can be moved to a plaquette ancilla qudit. Similarly any irrep error on a vertex whose star does not include other irrep errors can be moved to a vertex ancilla. 
We repeat these steps for multi-error configurations, by moving along paths on the dual or direct lattice respectively, returning the errored state to the code-space. We also discuss ribbon errors in App.~\ref{app:ribbon}. 
This results in an exact, deterministic recovery protocol. 
\end{enumerate}
We should stress that our protocol assumes as input neutral error clusters, and that we leave the problem of constructing an optimal syndrome extraction/decoder to identify these for separate work. 
In that sense, this work  serves as an exact quantum recovery step in a decoder, which is applicable to any finite group, including non-solvable ones, such as $A_5$. 
We leave an investigation of compatibility of our correction protocol with particular decoders and noise models, e.g. with the thresholds obtained from statistical-mechanics mappings \cite{Jing:2025zuu,Vadali:2026qon,Sala:2024uom,Sajith:toappear}, to future work.

\section{The Error Basis} \label{sec:errors}
\subsection{Non-Abelian Surface Code Setup}
We review the 2D surface code based on the Kitaev quantum double $D(G)$ \cite{Kitaev:1997wr}. Its local Hilbert space on each lattice edge is the group algebra
\be\label{eq:setup-01}
\bC[G]=\text{Span}\{\ket{g}:g\in G\}\,,
\ee 
with left and right multiplication operators acting as
\be \label{eq:LR_def}
L^h\ket{g}=\ket{hg}\,,\qquad R^h\ket{g}=\ket{gh^{-1}}\,,
\ee
for all $g,h\in G$. We also define the irreducible representation (irrep) operators
\be \label{eq:ZGamma_def}
Z^\Gamma_{ab}\ket{g}=\Gamma(g)_{ab}\ket{g}\,,
\ee
where $\Gamma\in\Irr(G)$, the set of irreducible representations of $G$, $d_\Gamma:=\dim\Gamma$, and $1\le a,b\le d_\Gamma$. We will take all representation matrices $\Gamma(g)$ to be unitary. 

The Hamiltonian for the model \cite{Kitaev:1997wr} is built from vertex and plaquette operators, which, on a square lattice, are:
\be\ba \label{eq:AvBp}
    A_v^{g}= \, & \begin{tikzpicture}[baseline]
\begin{scope}[shift={(0,0)}]
\draw[thick, ->-] (-1,0) to (0,0);
\draw[thick, ->-] (0,0) to (1,0);
\draw[thick, ->-] (0,-1) to (0,0);
\draw[thick, ->-] (0,0) to (0,1);
\node[above] at (0.2, 0) {$_v$};
\node[above] at (-0.7,0.05) {$R^g$};
\node[above] at (0.7,0.05) {$L^g$};
\node[right] at (0.05,0.7) {$L^g$};
\node[right] at (0.05,-0.7) {$R^g$};
\draw[fill=black] (0,0) ellipse (0.05 and 0.05);
\end{scope}
\end{tikzpicture}\,,\\
    B_p^{g}= & \sum\limits_{g_1,g_2,g_3,g_4}\delta_{g,\,g_1g_2g_3^{-1}g_4^{-1}} \times \\ 
{}&
\qquad 
\times  \begin{tikzpicture}[baseline]
\begin{scope}[shift={(0,0)}]
    \draw[->-] (-1,-0.5) to (-1,0.5);
    \node at (-1.5,0) {$\Bigg|$};
    \node at (-1.3,0) {$_{g_1}$};
    \draw[->-] (-1,0.5) to (0,0.5);
    \node at (-0.5,0.8) {$_{g_2}$};
    \draw[->-] (0,-0.5) to (0,0.5);
    \node at (0.3,0) {$_{g_3}$};
    \draw[->-] (-1,-0.5) to (0,-0.5);
    \node at (-0.5,-0.8) {$_{g_4}$};
    \node at (0.5,0) {$\Bigg>$};
    \node at (-0.5,0) {$_p$};
    \begin{scope}[shift={(2.2,0)}]
    \draw[->-] (-1,-0.5) to (-1,0.5);
    \node at (-1.5,0) {$\Bigg<$};
    \node at (-1.3,0) {$_{g_1}$};
    \draw[->-] (-1,0.5) to (0,0.5);
    \node at (-0.5,0.8) {$_{g_2}$};
    \draw[->-] (0,-0.5) to (0,0.5);
    \node at (0.3,0) {$_{g_3}$};
    \draw[->-] (-1,-0.5) to (0,-0.5);
    \node at (-0.5,-0.8) {$_{g_4}$};
    \node at (0.5,0) {$\Bigg|$};
    \node at (-0.5,0) {$_p$};        
    \end{scope}
\end{scope}    
\end{tikzpicture}\,.
\ea\ee
Our setup is a surface code for $D(G)$ on a finite size patch, with boundaries. Each boundary condition is labeled by a subgroup $K\leq G$ (here, we will restrict to trivial $2$-cocycles). The Hamiltonian for the bulk and boundary terms is then given by \cite{Kitaev:1997wr,Beigi:2010htr,Albert:2021vts,Cui:2019lvb,Huang:2025ump,Warman:2025hov}
\be \label{eq:Hamiltonian}
H=-\sum_vA_v-\sum_pB_p-\sum_{s}(A_{s}^{K_s}+B_{s}^{K_s})\,,
\ee
where the bulk operators are
\be\ba \label{eq:Av_Bp}
    A_v:=\frac{1}{|G|}\sum_g A_v^g\,,\qquad B_p:=B_p^{\id}\,,
\ea\ee
and the boundary ones are \cite{Beigi:2010htr}
\be\label{eq:setup-02}
    A_{s}^{K}:=\frac{1}{|K|}\sum_{k\in K} A_{s}^{k}\,,\qquad B_{s}^{K}:=\sum_{k\in K} B_{s}^{k}\,.
\ee
The first two sums run over bulk vertices $v$ and plaquettes $p$ while $s=(v,p)$ runs over the sites on the various boundaries. The truncated $A_{s}^{K_s}$ acts on the three edges at $v$, two along the boundary and one into the bulk, while $B_{s}^{K}$ acts on the single boundary edge $s$, with $B_{s}^k\ket{g}=\delta_{k,g}\ket{g}$.

\paragraf{Code space.} The code space is defined as the $+1$ eigenspace of the operators of the Hamiltonian \eqref{eq:Hamiltonian}: the logical states are its inequivalent ground states. Since the Hamiltonian terms are commuting projectors, each logical state $\ket{\psi}$ has to satisfy:
\be\ba
    A_v\ket{\psi}=B_p\ket{\psi}=A_{s}^{K_s}\ket{\psi}=B_{s}^{K_s}\ket{\psi}=\ket{\psi}\,,
\ea\ee
i.e. $\ket{\psi}$ is stabilized by the Hamiltonian terms.
The code space projector is then:
\be\label{eq:PAB}
P=\left(\prod_vA_v\right)\left(\prod_pB_p\right)\left(\prod_sA_s^{K_s}\right)\left(\prod_sB_s^{K_s}\right)\,.
\ee
and satisfies $P^2=P=P^\dagger$.

The logical operators are realized in terms of anyon ribbon operators that stretch between the top and bottom boundaries and a logical basis is generated by acting with them on the (unnormalized) reference state
\be \label{eq:KvAv}
\ol{\ket{\id}}=\left(\prod_vA_v\right)\left(\prod_sA_s^{K_s}\right)\bigotimes_i\ket{\id}_i\,.
\ee
where $i$ ranges over all lattice edges.
Setups on a square patch with opposite boundary conditions labeled by $K=\{\id\}$ (``rough'') and $K=G$ (``smooth'') were considered in \cite{Huang:2025ump,Manjunath:2026swm} and encode a logical group algebra $\mathbb{C}[G]$, while three different boundary conditions were used to encode a logical qubit for dihedral groups in \cite{Kobayashi:2025cfh,Warman:2025hov}.

\subsection{Construction of the Error Basis}\label{sec:errorbasis}

For an error basis, we can generalize the Pauli error basis for $\Z_2$ to the general finite group $G$ case. We will use the unitary left-multiplication operators $L^h$ for $h\in G$ as generalizations of Pauli $X$ and non-unitary rescaled irrep operators as generalizations of Pauli $Z$. Their products form an orthogonal basis of the operators on each single-site Hilbert space $\mathbb{C}[G]$ and their tensor-products span the operators on the physical Hilbert space $\bigotimes_i\mathbb{C}[G]_i$.

\begin{proposition}\label{prop:Emui_basis_single_site}
An orthogonal\footnote{We use the Hilbert–Schmidt inner product for linear operators: $\langle A,B\rangle_{\text{HS}}=\Tr(A^\dagger B)$ and denote by $\Tr_i$ the trace restricted to edge $i$.} basis of errors acting on each edge $i$ is
\be\label{eq:Emui}
E_i^{h,\mu}:=L_i^hE_i^\mu\,,\qquad E_i^\mu:=\sqrt{d_\Gamma}\,Z^\Gamma_{ab,i}\,.
\ee
Here $h\in G$ is a group element with $L^h$ defined in \eqref{eq:LR_def} while $\mu=(\Gamma,a,b)$, with $\Gamma\in\Irr(G)$ and $1\le a,b\le d_\Gamma$, labels the irrep operators defined in \eqref{eq:ZGamma_def}. They satisfy
\be\label{eq:Emunu_delta}
\Tr_i\!\left[(E_i^{h,\mu})^\dagger E_i^{k,\nu}\right]=|G|\,\delta_{h,k}\delta_{\mu,\nu}\,.
\ee
\end{proposition}

\begin{proof}
Since $L_i^g$ has diagonal entries which are all zero for $g\neq\id$, Schur orthogonality of the irreps \cite{Serre:1977} then gives
\be\label{eq:setup-10}\ba
\Tr_i\!\left[(L_i^hZ^\Gamma_{ab,i})^\dagger L_i^kZ^\Lambda_{cd,i}\right]&=\delta_{h,k}\sum_{g}\Gamma(g)_{ab}^*\Lambda(g)_{cd}\cr
&=\delta_{h,k}\frac{|G|}{d_\Gamma}\delta_{\Gamma\Lambda}\delta_{a,c}\delta_{b,d}\,,
\ea\ee
which, using the normalization of \eqref{eq:Emui}, is \eqref{eq:Emunu_delta}. Therefore, these $|G|\sum_\Gamma d_\Gamma^2=|G|^2$ orthogonal operators form a basis of the operators on each edge of the lattice.
\end{proof}

\paragraf{Multi-Qudit Errors.}
It is useful to introduce notation for the tensor product of errors on an ordered list of distinct edges $\vec i=(i_1,\ldots,i_N)$ 
\be\label{eq:multi-error-basis}\ba
E_{\vec i}^{\vec h,\vec\mu}&=L_{\vec i}^{\vec h}E_{\vec i}^{\vec\mu}=\prod_{n=1}^N E_{i_n}^{h_n,\mu_n}\,,\cr
L_{\vec i}^{\vec h}&=\prod_{n=1}^NL_{i_n}^{h_n}\,,\qquad E_{\vec i}^{\vec\mu}=\prod_{n=1}^NE_{i_n}^{\mu_n}\,,
\ea\ee
which, due to \eqref{eq:Emunu_delta}, obey the orthogonality condition 
\be\label{eq:multi-hs}
\Tr_{\vec i}\!\left[(E_{\vec i}^{\vec h,\vec\mu})^\dagger E_{\vec i}^{\vec k,\vec\nu}\right]=|G|^N\delta_{\vec h,\vec k}\delta_{\vec\mu,\vec\nu}\,.
\ee

\paragraf{Single-Qudit Irrep Basis.} Acting with the irrep errors on the uniform superposition of a single $G$-qudit gives the irrep basis states
\be\ba\label{eq:edge-irrep-basis}
\ket1&:=\frac{1}{\sqrt{|G|}}\sum_{g\in G}\ket g\,,\cr
\ket{\Gamma_{ab}}:=E^\mu\ket1&=\sqrt{\frac{d_\Gamma}{|G|}}\sum_{g\in G}\Gamma(g)_{ab}\ket g\,,
\ea\ee
with $\mu=(\Gamma,a,b)$. Since $E^\mu$ is diagonal in the group basis, $\braket{\Gamma_{ab}}{\Lambda_{cd}}=\frac{1}{|G|}\Tr\bigl[(E^\mu)^\dagger E^\nu\bigr]$, so these $|G|$ states are orthonormal by \eqref{eq:Emunu_delta}. {The irrep basis states are therefore linear combinations of the group basis states, and the change of basis is the non-Abelian Fourier transform. Since there are $\sum_\Gamma d_\Gamma^2=|G|$ irrep labels, we order them as $\mu_1=(1,1,1),\ldots,\mu_{|G|}$ and the group elements as $g_1=\id,\ldots,g_{|G|}$, and denote by $\ket{\Gamma,a,b}:=\ket{g_k}$ the group basis state in the same position as the irrep label $\mu_k=(\Gamma,a,b)$. The Fourier transform of a $G$-qudit is then the unitary
which in the group basis $F$ is the $|G|\times|G|$ matrix
\be\ba\label{eq:fourier-gate}
&F=\frac{1}{\sqrt{|G|}}\begin{pmatrix}1&\cdots&1\\\vdots&&\vdots\\\sqrt{d_\Gamma}\,\Gamma(g_1)^*_{ab}&\cdots&\sqrt{d_\Gamma}\,\Gamma(g_{|G|})^*_{ab}\\\vdots&&\vdots\end{pmatrix}\,,
\ea\ee
whose rows, labeled by $\mu=(\Gamma,a,b)$ starting with the trivial irrep, are the complex conjugated coefficients of the irrep states in the group basis, and whose columns are labeled by $g_1,\ldots,g_{|G|}$. By Schur orthogonality and the Peter-Weyl theorem, $F$ is unitary, $FF^\dagger=F^\dagger F=\bbI$, and it maps the irrep basis to the group basis.}

\subsection{Flux and Charge Checks}\label{sec:single-detect}
Isolated errors can be detected from the violations of the stabilizers. Indeed, when an isolated error carries a non-trivial multiplication label it produces non-trivial flux, i.e. a violation of the $B_p$ stabilizer, on its boundary plaquettes (and can produce charge at one vertex), while for non-trivial irrep labels it produces charge, detected by a violation of the $A_v$ stabilizer.

\paragraf{Violated Plaquettes and Vertices.} Consider an error as in \eqref{eq:Emui},
\be\ba\label{eq:general-error-recovery}
E_i^{h,\mu}&=L_i^hE_i^\mu=\sqrt{d_\Gamma}\,L_i^hZ^\Gamma_{ab,i}\,,\cr
h&\in G\,,\qquad\mu=(\Gamma,a,b)\,.
\ea\ee
For an edge $i$, let $p(i)$ and $v(i)$ be the plaquettes containing $i$ and the vertices incident on $i$ that carry a nontrivial constraint,
\be\ba
p(i)&:=\{p:B_p\neq\bbI\,,\ i\in\partial p\}\,,\cr
v(i)&:=\{v:A_v\neq\bbI\,,\ i\in\star v\}\,,
\ea\ee
where $\partial p:=\operatorname{supp}B_p$ and $\star v:=\operatorname{supp}A_v$. We call the vertices of $v(i)$ the {\bf constrained endpoints} of $i$. In the bulk $A_v$ and $B_p$ are never $\bbI$, but at the boundaries they can be, see \eqref{eq:setup-02}: $B_s^K=\bbI$ for $K=G$ (smooth) and $A_s^K=\bbI$ for $K=\{\id\}$ (rough). 

\begin{lemma}[Plaquette violations]\label{lemma:PBEPB}
Let $p\in p(i)$, and suppose that only the edge $i\in\partial p$ carries a non-trivial multiplication error $L_i^h$. Then, for every $h\neq\id$,
\be
B_pL_i^hB_p=0\,,\qquad B_pE_i^{h,\mu}B_p=0\,.
\ee
\end{lemma}
\begin{proof}
For a bulk plaquette with edges labeled as in \eqref{eq:AvBp}, $B_p$ imposes the trivial flux $\id=g_1g_2g_3^{-1}g_4^{-1}$. The error $L_i^h$ changes the flux depending on the site index $i$ as follows:
\be\ba\label{eq:L_flux_change}
L_1^h:\ \id&\mapsto(hg_1)g_2g_3^{-1}g_4^{-1}=h\,,\cr
L_2^h:\ \id&\mapsto g_1(hg_2)g_3^{-1}g_4^{-1}=g_1hg_1^{-1}\,,\cr
L_3^h:\ \id&\mapsto g_1g_2(hg_3)^{-1}g_4^{-1}=g_4h^{-1}g_4^{-1}\,,\cr
L_4^h:\ \id&\mapsto g_1g_2g_3^{-1}(hg_4)^{-1}=h^{-1}\,.
\ea\ee
The flux of the plaquette $p\in p(i)$ thus lies in the conjugacy classes $[h]$ or $[h^{-1}]$, which are both distinct from $[\id]=\{\id\}$ for $h\neq\id$, so $B_pL_i^hB_p=0$. Finally $Z^\Gamma_{ab,i}$ is diagonal in the group basis and preserves each flux sector, so $B_pE_i^{h,\mu}B_p=0$.
\end{proof}
The generalized plaquette operators \cite{Komar:2017uuc,Ritz-Zwilling:2025ptu}
\be \label{eq:Bp_conj_g}
    B_p^{[f]}:=\sum_{g\in[f]}B_p^g\,,
\ee
project onto the flux conjugacy class $[f]$ of $p$. The error $L_i^h$ moves the flux of a plaquette of $p(i)$ with no other multiplication errors from $[\id]$ to $[h]$ or $[h^{-1}]$, which is detected by $B_p^{[h]}$ or $B_p^{[h^{-1}]}$.

\paragraf{Vertex Violations.} The diagonal error $E_i^\mu$ commutes with every plaquette projector but is detected by the nontrivial vertex constraints on $v(i)$. We use the generalized vertex operators \cite{Komar:2017uuc,Ritz-Zwilling:2025ptu} with irrep labels $(\Gamma,a,b)$ for $\Gamma\in\Irr(G)$ and $a,b\in\{1,\cdots,d_\Gamma\}$.
\begin{proposition}\label{eq:prop_Avmu}
The generalized vertex operators
\be\label{eq:Avmu_def}
A_v^{(\Gamma,a,b)}:=\frac{d_\Gamma}{|G|}\sum_{g\in G}\Gamma(g)_{ab}A_v^g\,,
\ee
satisfy the matrix-unit algebra
\be\label{eq:Avmu_comm}
A_v^{(\Gamma,a,b)}A_v^{(\Lambda,c,d)}=\delta_{\Gamma,\Lambda}\delta_{b,c}A_v^{(\Gamma,a,d)}\,,
\ee
and
\be\label{eq:Avmu-adjoint}
\bigl(A_v^{(\Gamma,a,b)}\bigr)^\dagger=A_v^{(\Gamma,b,a)}\,.
\ee
\end{proposition}
This is proven in App.~\ref{app:proof2}.

In this notation, the standard vertex projector $A_v$ \eqref{eq:Av_Bp} is labeled by the trivial irrep $(\Gamma,a,b)=(1,1,1)$:
\be
    A_v^{(1,1,1)}=A_v\,.
\ee
\begin{lemma}[Vertex violations]\label{lemma:vertex-violations}
Let $v\in v(i)$, and suppose that $i$ is the only edge of $\star v$ carrying an error. Then, for every $\mu=(\Gamma,a,b)$ with $\Gamma\neq1$,
\be\label{eq:AvEAv}
A_vE_i^\mu A_v=0\,,
\ee
so every nontrivial vertex constraint in $v(i)$ is violated.
\end{lemma}
\begin{proof}
If $i$ is oriented away from $v$, then $A_v^g$ acts by $L_i^g$, and \eqref{eq:ZL_comm} gives
\be
E_i^{(\Gamma,a,b)}A_v=\frac{1}{d_\Gamma}\sum_{c=1}^{d_\Gamma}A_v^{(\Gamma,a,c)}E_i^{(\Gamma,c,b)}\,.
\ee
while if $i$ is oriented towards $v$, $A_v^g$ acts on $i$ by $R_i^g$, and \eqref{eq:ZR_comm} gives
\be
E_i^{(\Gamma,a,b)}A_v=\frac{1}{d_\Gamma}\sum_{c=1}^{d_\Gamma}A_v^{(\Gamma^{*},b,c)}E_i^{(\Gamma,a,c)}\,,
\ee
In either case, left multiplication by $A_v=A_v^{(1,1,1)}$ annihilates the right-hand side for $\Gamma\neq1$ by the matrix-unit algebra \eqref{eq:Avmu_comm}.
\end{proof}

\paragraf{Induced Charge Errors.} A multiplication error $L_i^h$ can also create charge, but only at its source endpoint when that endpoint has a non-trivial vertex check $A_v$. We will denote by $v_-,v_+$ the vertices at the endpoints of $i$, with $i$ oriented from $v_-$ to $v_+$:
\be
\begin{split}
     \begin{tikzpicture}[scale=1.5,xscale=1]
        \draw[step=0.5,black] (1,0.5) grid (2.5,1.5);
		\begin{scope}[shift={(1,0)}]
		\draw[white, thick] (1.5,0) -- (1.5,2);
		\draw[red, thick] (0.5,1) -- (1,1);
		\node[red] at (0.75,1.2) {$i$};
		\node[black] at (1.2,0.9) {${}_{v_+}$};
        \node[black] at (0.3,0.9) {${}_{v_-}$};
		\end{scope}
	\end{tikzpicture}
\end{split}
\ee
$A_{v_-}^g$ acts on $i$ by left multiplication and the error $L_i^h$ acting on a code state $\ket\psi$ obeys
\be\label{eq:Lh-conjugation}
A_{v_-}^gL_i^h\ket\psi=L_i^{ghg^{-1}}A_{v_-}^g\ket\psi=L_i^{ghg^{-1}}\ket\psi\,.
\ee
The centralizer $C_G(h):=\{g\in G:gh=hg\}$ therefore acts trivially at $v_-$, consistently with the anyon label $([h],1)$, in which $1$ is the trivial irrep of $C_G(h)$. As a representation of $G$, the charge at $v_-$ is the induced representation \cite{Komar:2017uuc}
\begin{align}\label{eq:induced-charge}
\Gamma_{[f]}&:=\operatorname{Ind}_{C_G(f)}^G1\cong\bigoplus_{\Gamma\in\Irr(G)}n_\Gamma^{[f]}\,\Gamma\,,\\
n_\Gamma^{[f]}&=\bigl\langle\operatorname{Res}_{C_G(f)}^G\Gamma,1\bigr\rangle_{C_G(f)}=\frac1{|C_G(f)|}\sum_{g\in C_G(f)}\chi_\Gamma(g)^*\,,\nn
\end{align}
with multiplicities from Frobenius reciprocity and $\chi_\Gamma(g)=\Tr[\Gamma(g)]$.\footnote{Since conjugate elements have conjugate centralizers, $\Gamma_{[f]}$ depends only on the conjugacy class.} 
For non-central $h$ it contains nontrivial irreps. At the other endpoint $v_+$, $A_{v_+}^g$ acts by $R_i^g$, which commutes with $L_i^h$, so no induced charge is created at $v_+$.

\subsection{Boundary Checks and Errors}

We now consider errors near a boundary labeled by a subgroup $K\leq G$, as in (\ref{eq:setup-02}). 
Since $A_v^KA_v^k=A_v^K$ for $k\in K$, the code-space conditions $A_v^KP=P$ and $B_s^KP=P$ at a boundary vertex $v$ and a boundary edge $s$ are equivalent to
\be\label{eq:boundary-stabilizers}
A_v^kP=P\quad(k\in K)\,,\qquad B_s^h P=0\quad(h\notin K)\,,
\ee
i.e.\ only the elements of $K$ act as stabilizers at the boundary. The boundary checks therefore detect errors only up to elements of $K$, and Lemmas~\ref{lemma:PBEPB} and \ref{lemma:vertex-violations} are modified as follows. The two limiting cases are $K=G$ (smooth), for which $B_s^G=\bbI$ and $A_v^G$ is the three-edge version of $A_v$, and $K=\{\id\}$ (rough), for which $A_v^{\{\id\}}=\bbI$ and $B_s^{\id}$ pins $s$ to $\ket\id$.

\paragraf{Multiplication Errors.} For $g\in K$ one has $hg\in K$ if and only if $h\in K$, so for a boundary edge $s$
\be\ba\label{eq:boundary-flux}
B_s^KL_s^hB_s^K&=L_s^hB_s^K\quad(h\in K)\,,\cr
B_s^KL_s^hB_s^K&=0\quad(h\notin K)\,.
\ea\ee
This is due to the fact that the fluxes $[k]$, $k\in K$, condense on the boundary \cite{Beigi:2010htr}. 
The boundary site $s$ also lies on a bulk four-edge plaquette $p$, which detects $L_s^k$ for every $k\neq\id$ by Lemma~\ref{lemma:PBEPB} if it is the only multiplication error acting on $\de p$. A {single} $L_s^k$ on the boundary labeled by subgroup $K$ is therefore still a detectable error.
Instead, if $L^k$ is the termination of an anyon ribbon, the bulk plaquette will have another multiplication error and trivial total flux, as follows from the anyon ribbon properties and  and will be detected in its bulk endpoint as we discuss later  in Sec. \ref{sec:RibEr}.

\paragraf{Irrep Operators on Boundary Edges.} 
On the code space, the variable of a boundary edge $s$ lies in $K$. By Frobenius reciprocity, we define
\be\label{eq:nGammaK}
n_\Gamma^K:=\bigl\langle\operatorname{Res}^G_K\Gamma,1\bigr\rangle_K=\bigl\langle\Gamma,\operatorname{Ind}_K^G1\bigr\rangle_G\,,
\ee 
and we choose the basis of $\Gamma$ such that its first $n_\Gamma^K$ vectors span the $K$-invariant subspace. For $k\in K$ the matrix elements $\Gamma(k)_{ab}$ with $a\le n_\Gamma^K$ or $b\le n_\Gamma^K$ equal $\delta_{a,b}$, since $\Gamma(k)$ is unitary and fixes the first $n_\Gamma^K$ basis vectors. The corresponding operators therefore act as scalars on the code space,
\be\label{eq:boundary-edge-collapse}
E_s^{(\Gamma,a,b)}P={\sqrt{d_\Gamma}\,\delta_{a,b}P}\qquad(a\le n_\Gamma^K\ \text{or}\ b\le n_\Gamma^K)\,,
\ee
A special instance is {the} rough boundary, in which $K=\{\id\}$, {$n_\Gamma^{\{\id\}}=d_\Gamma$} and every irrep label acts trivially. More generally, $E_s^{\mu}P$ only depends on $\operatorname{Res}^G_K\Gamma$. In a basis of $\Gamma$ adapted to the decomposition of $\operatorname{Res}^G_K\Gamma$ into irreps of $K$, each $E_s^{(\Gamma,a,b)}P$ is either zero or proportional to one of the $\sum_{{\Gamma_K}\in\Irr(K)}d_{\Gamma_K}^2=|K|$ operators
\be\ba\label{eq:boundary-K-irrep}
E_s^{({\Gamma_K},a,b)}P&:=\sqrt{d_{\Gamma_K}}\,Z^{\Gamma_K}_{ab,s}P\,,\cr
Z^{\Gamma_K}_{ab}\ket{g}&={\Gamma_K}(g)_{ab}\ket{g}\quad(g\in K)\,,
\ea\ee
with ${\Gamma_K}\in\Irr(K)$. The $|G|$ irrep labels on an edge on a boundary labeled by the subgroup $K$ therefore collapse to $|K|$ labels on the code space. By Lemma~\ref{lemma:vertex-violations} with $G$ replaced by $K$ and $\Gamma$ replaced by ${\Gamma_K}$, the nontrivial $K$-irreps are detected by the boundary checks $A_v^K$.

\subsection{Orthogonality on the Code Space}\label{sec:KL}

\paragraf{Check Region.} For each edge $i$, we define the check region as the support of the vertex and plaquette checks acting on $i$:
\be\label{eq:single-local-region}
c(i):=\bigcup_{v\in v(i)}\star v\ \cup\bigcup_{p\in p(i)}\partial p\,,
\ee
which is shown in \eqref{eq:check-region} for a bulk horizontal edge, with the edge $i$ in red, the two plaquettes of $p(i)$ shaded in blue and the other edges of the check region, the stars of $v(i)$ and the boundaries of $p(i)$, in thick blue (the check region also includes the red edge $i$): 
\be 
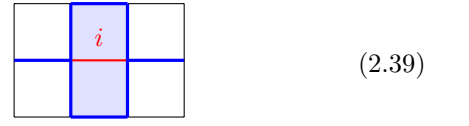
\label{eq:check-region}
\begin{split}
\begin{tikzpicture}[scale=1.5,xscale=1]
        \draw[step=0.5,black] (0,0.5) grid (1.5,1.5);
          \fill[blue!12] (0.5,0.5)-- (1,0.5) -- (1,1.5) --(0.5,1.5) -- (0.5,0.5) ;
		\draw[red, thick] (0.5,1) -- (1,1);
        \draw[blue, very thick] (0,1) -- (0.5,1);
		\draw[blue, very thick] (1,1) -- (1.5,1);
        \draw[blue, very thick] (0.5,0.5) -- (0.5,1.5);
	\draw[blue, very thick] (1,0.5) -- (1,1.5);
		\draw[blue, very thick] (0.5,1.5) -- (1,1.5);
		\draw[blue, very thick] (0.5,0.5) -- (1,0.5);
		\node[red] at (0.75,1.2) {${i}$};
	\end{tikzpicture}
\end{split}
\ee
It supports the local checks that act on the edge $i$ and the local code-space projector $P_{c(i)}$:
\be\ba\label{eq:local-checks}
A_{v(i)}&:=\prod_{v\in v(i)}A_v\,,\qquad B_{p(i)}:=\prod_{p\in p(i)}B_p\,,\cr
P_{c(i)}&:=A_{v(i)}B_{p(i)}\,.
\ea\ee
For multiple edges $\vec i=(i_1,\ldots,i_N)$ with, the local checks \eqref{eq:local-checks} for each error cluster are their products
\be\ba\label{eq:multi-local-checks}
v(\vec i)&:=\bigcup_{n=1}^Nv(i_n)\,,\qquad p(\vec i):=\bigcup_{n=1}^Np(i_n)\,,\cr
A_{v(\vec i)}&:=\prod_{v\in v(\vec i)}A_v\,,\quad B_{p(\vec i)}:=\prod_{p\in p(\vec i)}B_p\,,\cr
P_{c(\vec i)}&:=A_{v(\vec i)}B_{p(\vec i)}\,,\quad c(\vec i):=\bigcup_{v\in v(\vec i)}\star v\ \cup\bigcup_{p\in p(\vec i)}\partial p\,,
\ea\ee
with $c(\vec i)$ defined from $v(\vec i)$ and $p(\vec i)$ as in \eqref{eq:single-local-region}. 

\paragraf{Stabilizers and Redundancy.} Each edge carries $|G|$ multiplication labels and $\sum_\Gamma d_\Gamma^2=|G|$ irrep labels. Thus a star or plaquette supports $|G|^{2\times4}=|G|^8$
linearly independent operators. However, on the code space this count is reduced by a factor $|G|$, owing to the vertex or plaquette stabilizers, which form loops on the dual or direct lattice. The stabilizers are the $|G|$ operators $A_v^g$ for $g\in G$ and the linear combinations $D_p$ of the $|G|$ operators $B_p^g$ such that $D_pP=P$. We will now discuss gauge-fixing for both vertex and plaquette stabilizers to obtain an orthogonal error family on the code-space. This is depicted in Fig.~\ref{fig:gauge-fix}.

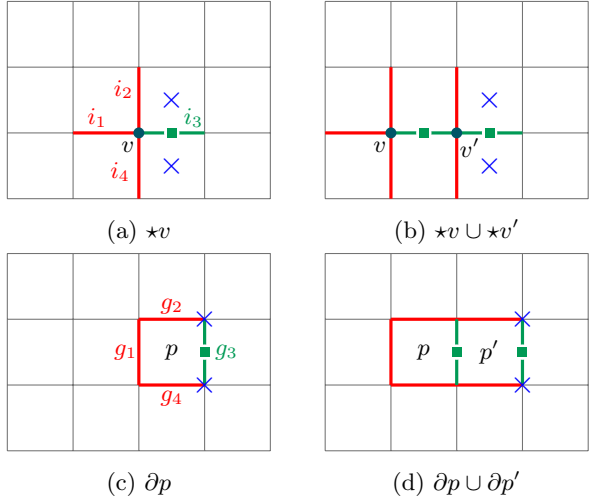
\begin{figure}
\centering
\begin{tikzpicture}[scale=1.45]
\begin{scope}[shift={(0,0)}]
  \draw[step=0.6,black!55] (0,0) grid (2.4,1.8);
  \draw[red,very thick] (0.6,0.6) -- (1.2,0.6);
  \draw[red,very thick] (1.2,0.6) -- (1.2,1.2);
  \draw[red,very thick] (1.2,0) -- (1.2,0.6);
  \draw[ForestGreen,very thick] (1.2,0.6) -- (1.8,0.6);
  \fill[blue(munsell)] (1.2,0.6) circle (0.05);
  \fill[white] (1.5,0.6) circle (0.07);
  \fill[ForestGreen] (1.5,0.6) +(-0.045,-0.045) rectangle +(0.045,0.045);
  \node[blue,font=\large] at (1.5,0.9) {$\times$};
  \node[blue,font=\large] at (1.5,0.3) {$\times$};
  \node[red] at (0.82,0.74) {${i_1}$};
  \node[red] at (1.06,1.0) {${i_2}$};
  \node[red] at (1.03,0.25) {${i_4}$};
  \node[ForestGreen] at (1.7,0.74) {${i_3}$};
  \node at (1.1,0.47) {${v}$};
  \node at (1.2,-0.3) {(a) $\star v$};
\end{scope}
\begin{scope}[shift={(2.9,0)}]
  \draw[step=0.6,black!55] (0,0) grid (2.4,1.8);
  \draw[red,very thick] (0,0.6) -- (0.6,0.6);
  \draw[red,very thick] (0.6,0.6) -- (0.6,1.2);
  \draw[red,very thick] (0.6,0) -- (0.6,0.6);
  \draw[red,very thick] (1.2,0.6) -- (1.2,1.2);
  \draw[red,very thick] (1.2,0) -- (1.2,0.6);
  \draw[ForestGreen,very thick] (0.6,0.6) -- (1.2,0.6);
  \draw[ForestGreen,very thick] (1.2,0.6) -- (1.8,0.6);
  \fill[blue(munsell)] (0.6,0.6) circle (0.05);
  \fill[blue(munsell)] (1.2,0.6) circle (0.05);
  \fill[white] (0.9,0.6) circle (0.07);
  \fill[ForestGreen] (0.9,0.6) +(-0.045,-0.045) rectangle +(0.045,0.045);
  \fill[white] (1.5,0.6) circle (0.07);
  \fill[ForestGreen] (1.5,0.6) +(-0.045,-0.045) rectangle +(0.045,0.045);
  \node[blue,font=\large] at (1.5,0.9) {$\times$};
  \node[blue,font=\large] at (1.5,0.3) {$\times$};
  \node at (0.5,0.47) {${v}$};
  \node at (1.32,0.47) {${v'}$};
  \node at (1.2,-0.3) {(b) $\star v\cup\star v'$};
\end{scope}
\begin{scope}[shift={(0,-2.3)}]
  \draw[step=0.6,black!55] (0,0) grid (2.4,1.8);
  \draw[red,very thick] (1.2,1.2) -- (1.8,1.2);
  \draw[ForestGreen,very thick] (1.8,0.6) -- (1.8,1.2);
  \draw[red,very thick] (1.2,0.6) -- (1.8,0.6);
  \draw[red,very thick] (1.2,0.6) -- (1.2,1.2);
  \fill[white] (1.8,0.9) circle (0.07);
  \fill[ForestGreen] (1.8,0.9) +(-0.045,-0.045) rectangle +(0.045,0.045);
  \node[blue,font=\large] at (1.8,1.2) {$\times$};
  \node[blue,font=\large] at (1.8,0.6) {$\times$};
  \node[red] at (1.08,0.9) {${g_1}$};
  \node[red] at (1.5,1.32) {${g_2}$};
  \node[ForestGreen] at (2,0.9) {${g_3}$};
  \node[red] at (1.5,0.48) {${g_4}$};
  \node at (1.5,0.9) {${p}$};
  \node at (1.2,-0.3) {(c) $\partial p$};
\end{scope}
\begin{scope}[shift={(2.9,-2.3)}]
  \draw[step=0.6,black!55] (0,0) grid (2.4,1.8);
  \draw[red,very thick] (0.6,1.2) -- (1.2,1.2);
  \draw[red,very thick] (1.2,1.2) -- (1.8,1.2);
  \draw[ForestGreen,very thick] (1.8,0.6) -- (1.8,1.2);
  \draw[red,very thick] (0.6,0.6) -- (1.2,0.6);
  \draw[red,very thick] (1.2,0.6) -- (1.8,0.6);
  \draw[red,very thick] (0.6,0.6) -- (0.6,1.2);
  \draw[ForestGreen,very thick] (1.2,0.6) -- (1.2,1.2);
  \fill[white] (1.2,0.9) circle (0.07);
  \fill[ForestGreen] (1.2,0.9) +(-0.045,-0.045) rectangle +(0.045,0.045);
  \fill[white] (1.8,0.9) circle (0.07);
  \fill[ForestGreen] (1.8,0.9) +(-0.045,-0.045) rectangle +(0.045,0.045);
  \node[blue,font=\large] at (1.8,1.2) {$\times$};
  \node[blue,font=\large] at (1.8,0.6) {$\times$};
  \node at (0.9,0.9) {${p}$};
  \node at (1.5,0.9) {${p'}$};
  \node at (1.2,-0.3) {(d) $\partial p\cup\partial p'$};
\end{scope}
\end{tikzpicture}
\caption{Errors are shown in red, gauge-fixed edges are shown in green and the single-error checks are shown in blue. We gauge-fix the right edge of each vertex star to have $h_3=\id$ and the right edge of each plaquette to have $\mu_3=1$.}
\label{fig:gauge-fix}
\end{figure}

\paragraf{Vertex Gauge Fixing.} On each $\star v$, the stabilizers are
\be\ba\label{eq:Avg-factorized}
A_v^g&=R_{i_1}^gL_{i_2}^gL_{i_3}^gR_{i_4}^g\,,\quad A_v^gP&=P\,,
\ea\ee
and a gauge-fixed error family is given by
\be\label{eq:vertex-fix}
\cE_{\star v}:=\{E_{\star v}^{\vec h,\vec\mu}:h_n=\id\}\,,\quad
|\cE_{\star v}|=|G|^7\,,
\ee
where we fix an edge $i_n$ that leaves vertex $v$.
Fixing the multiplication error label $h_n=\id$ makes every retained error diagonal on $i_n$. Since $L_{i_n}^g$ has no diagonal entries unless $g=\id$, the gauge-fixed family \eqref{eq:vertex-fix} cannot contain any non-trivial vertex stabilizers \eqref{eq:Avg-factorized}.  Conversely, if the multiplication error on edge $i_n$ carries the label $h_n$, then multiplying the error on the right by the stabilizer $A_v^{h_n^{-1}}$ sets this label to $\id$, since $L_{i_n}^{h_n}E_{i_n}^{\mu_n}L_{i_n}^{h_n^{-1}}$ is diagonal by \eqref{eq:ZL_comm}. {On the other edges of $\star v$, $A_v^{h_n^{-1}}$ acts by $L^{h_n^{-1}}$ or $R^{h_n^{-1}}$, and since $R^g=\sum_{k\in G}L^k\Pi^g_k$ with the diagonal projectors $\Pi^g_k$ onto $\{x\in G:xg^{-1}=kx\}$, the result is again a linear combination of the basis errors \eqref{eq:multi-error-basis}, now with $h_n=\id$.} This shows that the original error differs by the gauge-fixed one by a stabilizer and is therefore equivalent to it when acting on the code-space. For neighboring vertices, we can gauge-fix e.g. from right to left, see Fig.~\ref{fig:gauge-fix}(b).

\paragraf{Plaquette Gauge Fixing} For each plaquette, we fix the irrep label on an edge $i_n$ and define the error family:
\be\label{eq:plaquette-fix}
\cE_{\partial p}:=\{E_{\partial p}^{\vec h,\vec\mu}:\mu_{n}=1\}\,,\quad
|\cE_{\partial p}|=|G|^7\,.
\ee
The stabilizer condition $D_pP=P$ imposes the flatness condition $g_1g_2g_3^{-1}g_4^{-1}=\id$, so, using the property of representations, any irrep error on one edge can be expressed as linear combinations of errors on the other three edges, e.g. for $\mu=(\Gamma,a,b)$, and $i_n=i_3$
\be\label{eq:plaquette-move}
E_{i_3}^{\mu}P=\sqrt{d_\Gamma}\sum_{c,d=1}^{d_\Gamma}Z^{\Gamma^*}_{ca,i_4}Z^\Gamma_{cd,i_1}Z^\Gamma_{db,i_2}P\,.
\ee
For each pair of neighboring plaquettes, we will use their two flatness constraints to gauge-fix two edges e.g. we gauge fix the right edge of each plaquette, moving from right to left as in Fig.~\ref{fig:gauge-fix}(d).

{\paragraf{Completed Supports.} An error-free edge carries the trivial labels $h=\id$ and $\mu=1$, so we include it in the support as an effectively gauge-fixed edge. E.g. for errors on the six edges of $\star v_-\cup\star v_+$ other than the shared edge $i_{v_-v_+}$, see Fig.~\ref{fig:gauge-fix-protocol}(a), these six edges form a dual loop without a complete star, and $A_{v_-}^hA_{v_+}^h$ would make the errors non-orthogonal on the code space. Including $i_{v_-v_+}$ as the fixed edge of $\star v_-$ gives the two-star support of Fig.~\ref{fig:gauge-fix}(b), and similarly for an error-free edge shared by two plaquettes, see Fig.~\ref{fig:gauge-fix}(d). We always take $\vec i$ to be completed in this way, so that gauge fixing removes every direct and dual loop in $\vec i$.}

\paragraf{Gauge-fixed Error Family.} On a given support of edges $\vec i=(i_1,\cdots,i_N)$, we will implement the vertex gauge-fixing \eqref{eq:vertex-fix} and denote the fixed edges by $\vec i_A$ and the plaquette gauge-fixing \eqref{eq:plaquette-fix} denoting the fixed edges by $\vec i_B$. The resulting gauge-fixed error basis is then:
\be\ba\label{eq:gauge-fixed-family}
\cE_{\vec i}:=\bigl\{E_{\vec i}^{\vec h,\vec\mu}:\ &h_n=\id\quad(i_n\in\vec i_A)\,,\cr
&\mu_n=1\quad(i_n\in\vec i_B)\bigr\}\,,
\ea\ee
with
\be\label{eq:gauge-fixed-count}
|\cE_{\vec i}|=|G|^{2N-|\vec i_A|-|\vec i_B|}.
\ee
Indeed, each complete star (a dual lattice cycle), removes one factor $|G|$ of multiplication labels, and each complete plaquette (a direct cycle) one factor $|G|$ of irrep labels. This provides a complete error basis, since any multiplication error on an edge in $\vec i_A$ or irrep label on $\vec i_B$ can be written using the stabilizer properties, as we showed above.

\begin{figure*}
\centering
\begin{tikzpicture}[scale=1.45]
\begin{scope}[shift={(0,0)}]
  \draw[step=0.6,black!55] (0,0) grid (2.4,1.8);
  \draw[red,very thick] (0,0.6) -- (0.6,0.6);
  \draw[red,very thick] (0.6,0.6) -- (0.6,1.2);
  \draw[red,very thick] (0.6,0) -- (0.6,0.6);
  \draw[red,very thick] (1.2,0.6) -- (1.2,1.2);
  \draw[red,very thick] (1.2,0) -- (1.2,0.6);
  \draw[red,very thick] (0.6,0.6) -- (1.2,0.6);
  \draw[red,very thick] (1.2,0.6) -- (1.8,0.6);
  \fill[blue(munsell)] (0.6,0.6) circle (0.05);
  \fill[blue(munsell)] (1.2,0.6) circle (0.05);
  \node at (0.47,0.47) {${}_{v_-}$};
  \node at (1.33,0.47) {${}_{v_+}$};
  \node at (1.2,-0.3) {{(a)} seven errors};
\end{scope}
\begin{scope}[shift={(2.9,0)}]
  \draw[step=0.6,black!55] (0,0) grid (2.4,1.8);
  \draw[red,very thick] (0,0.6) -- (0.6,0.6);
  \draw[red,very thick] (0.6,0.6) -- (0.6,1.2);
  \draw[red,very thick] (0.6,0) -- (0.6,0.6);
  \draw[red,very thick] (1.2,0.6) -- (1.2,1.2);
  \draw[red,very thick] (1.2,0) -- (1.2,0.6);
  \draw[red,very thick] (0.6,0.6) -- (1.2,0.6);
  \draw[ForestGreen,very thick] (1.2,0.6) -- (1.8,0.6);
  \fill[blue(munsell)] (0.6,0.6) circle (0.05);
  \fill[blue(munsell)] (1.2,0.6) circle (0.05);
  \fill[white] (1.5,0.6) circle (0.07);
  \fill[ForestGreen] (1.5,0.6) +(-0.045,-0.045) rectangle +(0.045,0.045);
  \node[draw,circle,inner sep=0.6pt,font=\tiny] at (1.55,0.36) {1};
  \draw[->,thick,blue(munsell)] (1.5,0.72) to[bend right=50] (0.9,0.72);
  \node at (0.47,0.47) {${}_{v_-}$};
  \node at (1.33,0.47) {${}_{v_+}$};
  \node at (1.2,-0.3) {{(b)} gauge-fix at $v_+$};
\end{scope}
\begin{scope}[shift={(5.8,0)}]
  \draw[step=0.6,black!55] (0,0) grid (2.4,1.8);
  \draw[red,very thick] (0,0.6) -- (0.6,0.6);
  \draw[red,very thick] (0.6,0.6) -- (0.6,1.2);
  \draw[red,very thick] (0.6,0) -- (0.6,0.6);
  \draw[red,very thick] (1.2,0.6) -- (1.2,1.2);
  \draw[red,very thick] (1.2,0) -- (1.2,0.6);
  \draw[ForestGreen,very thick] (0.6,0.6) -- (1.2,0.6);
  \draw[ForestGreen,very thick] (1.2,0.6) -- (1.8,0.6);
  \fill[blue(munsell)] (0.6,0.6) circle (0.05);
  \fill[blue(munsell)] (1.2,0.6) circle (0.05);
  \fill[white] (0.9,0.6) circle (0.07);
  \fill[ForestGreen] (0.9,0.6) +(-0.045,-0.045) rectangle +(0.045,0.045);
  \fill[white] (1.5,0.6) circle (0.07);
  \fill[ForestGreen] (1.5,0.6) +(-0.045,-0.045) rectangle +(0.045,0.045);
  \node[draw,circle,inner sep=0.6pt,font=\tiny] at (1.55,0.36) {1};
  \node[draw,circle,inner sep=0.6pt,font=\tiny] at (0.9,0.36) {2};
  \draw[->,thick,blue(munsell)] (0.9,0.72) to[bend right=50] (0.3,0.72);
  \node at (0.47,0.47) {${}_{v_-}$};
  \node at (1.33,0.47) {${}_{v_+}$};
  \node at (1.2,-0.3) {{(c)} gauge-fix at $v_-$};
\end{scope}
\begin{scope}[shift={(0,-2.3)}]
  \draw[step=0.6,black!55] (0,0) grid (2.4,1.8);
  \draw[red,very thick] (0.6,0.6) -- (0.6,1.2);
  \draw[red,very thick] (0.6,1.2) -- (1.2,1.2);
  \draw[red,very thick] (0.6,0.6) -- (1.2,0.6);
  \draw[red,very thick] (1.2,1.2) -- (1.8,1.2);
  \draw[red,very thick] (1.2,0.6) -- (1.8,0.6);
  \draw[red,very thick] (1.2,0.6) -- (1.2,1.2);
  \draw[red,very thick] (1.8,0.6) -- (1.8,1.2);
  \node at (0.9,0.75) {${p}$};
  \node at (1.62,0.75) {${p'}$};
  \node at (1.2,-0.3) {{(d)} seven errors};
\end{scope}
\begin{scope}[shift={(2.9,-2.3)}]
  \draw[step=0.6,black!55] (0,0) grid (2.4,1.8);
  \draw[red,very thick] (0.6,0.6) -- (0.6,1.2);
  \draw[red,very thick] (0.6,1.2) -- (1.2,1.2);
  \draw[red,very thick] (0.6,0.6) -- (1.2,0.6);
  \draw[red,very thick] (1.2,1.2) -- (1.8,1.2);
  \draw[red,very thick] (1.2,0.6) -- (1.8,0.6);
  \draw[red,very thick] (1.2,0.6) -- (1.2,1.2);
  \draw[ForestGreen,very thick] (1.8,0.6) -- (1.8,1.2);
  \fill[white] (1.8,0.9) circle (0.07);
  \fill[ForestGreen] (1.8,0.9) +(-0.045,-0.045) rectangle +(0.045,0.045);
  \node[draw,circle,inner sep=0.6pt,font=\tiny] at (1.98,0.9) {1};
  \draw[->,thick,blue(munsell)] (1.72,1.02) to[bend right=50] (1.28,1.02);
  \node at (0.9,0.75) {${p}$};
  \node at (1.62,0.75) {${p'}$};
  \node at (1.2,-0.3) {{(e)} gauge-fix at $p'$};
\end{scope}
\begin{scope}[shift={(5.8,-2.3)}]
  \draw[step=0.6,black!55] (0,0) grid (2.4,1.8);
  \draw[red,very thick] (0.6,0.6) -- (0.6,1.2);
  \draw[red,very thick] (0.6,1.2) -- (1.2,1.2);
  \draw[red,very thick] (0.6,0.6) -- (1.2,0.6);
  \draw[red,very thick] (1.2,1.2) -- (1.8,1.2);
  \draw[red,very thick] (1.2,0.6) -- (1.8,0.6);
  \draw[ForestGreen,very thick] (1.2,0.6) -- (1.2,1.2);
  \draw[ForestGreen,very thick] (1.8,0.6) -- (1.8,1.2);
  \fill[white] (1.2,0.9) circle (0.07);
  \fill[ForestGreen] (1.2,0.9) +(-0.045,-0.045) rectangle +(0.045,0.045);
  \fill[white] (1.8,0.9) circle (0.07);
  \fill[ForestGreen] (1.8,0.9) +(-0.045,-0.045) rectangle +(0.045,0.045);
  \node[draw,circle,inner sep=0.6pt,font=\tiny] at (1.98,0.9) {1};
  \node[draw,circle,inner sep=0.6pt,font=\tiny] at (1.36,0.72) {2};
  \draw[->,thick,blue(munsell)] (1.12,1.02) to[bend right=50] (0.68,1.02);
  \node at (0.9,0.75) {${p}$};
  \node at (1.62,0.75) {${p'}$};
  \node at (1.2,-0.3) {{(f)} gauge-fix at $p$};
\end{scope}
\end{tikzpicture}
\caption{{The gauge-fixing protocol, with the conventions of Fig.~\ref{fig:gauge-fix}. Errors are shown in red, gauge-fixed edges are green and the numbers give the order in which their gauge-fixed labels are set, proceeding along each row from right to left. Top row: (a) errors on the seven edges of $\star v_-\cup\star v_+$. (b) The stabilizer $A_{v_+}^{h^{-1}}$ sets the multiplication label of the right edge of $v_+$ to $\id$, and moves an error onto $i_{v_-v_+}$. (c) $A_{v_-}^{h^{-1}}$ then sets the multiplication label of $i_{v_-v_+}$ to $\id$, moving the error onto edges of $\star v_-$ that already carry errors, which gives the gauge-fixed two-star support of Fig.~\ref{fig:gauge-fix}(b). Bottom row: (d) errors on the seven edges of $\partial p\cup\partial p'$. (e) The flatness of $p'$, \eqref{eq:plaquette-move}, sets the irrep label of the right edge of $p'$ to $1$, and moves an irrep error onto the shared edge. (f) The flatness of $p$ then sets the irrep label of the shared edge to $1$, moving the error onto edges of $\partial p$, which gives the gauge-fixed two-plaquette support of Fig.~\ref{fig:gauge-fix}(d).}}
\label{fig:gauge-fix-protocol}
\end{figure*}
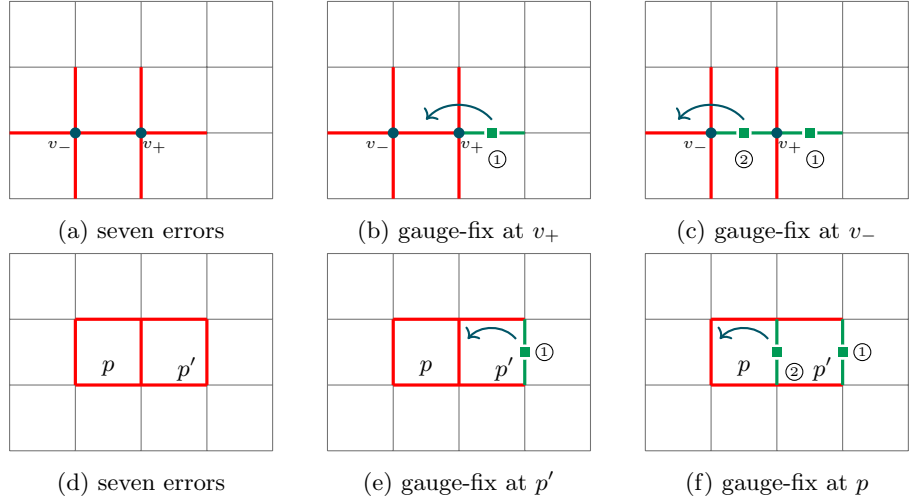

\begin{lemma}[Completeness of the Gauge-Fixed Error family] Any error can be brought into gauge-fixed form, i.e. it is equivalent up to stabilizers to a linear combination of errors in \eqref{eq:gauge-fixed-family}.
\end{lemma}
\begin{proof}
{We will re-define a general error by multiplying it with stabilizers that act trivially on the code-space and bring it to gauge-fixed form proceeding along each row of the lattice from right to left, see Fig.~\ref{fig:gauge-fix-protocol}.}

{\paragraf{Vertex step.} At the right-most complete star $\star v_+$ of a row (which received the edge $i_n$ with group element $h_n$), we multiply the error by $A_{v_+}^{h_n^{-1}}$, which sets $h_n=\id$ as shown below \eqref{eq:vertex-fix}. This operator also acts on the edge $i_{v_-v_+}$ entering $v_+$, which can now carry an error, see Fig.~\ref{fig:gauge-fix-protocol}(b). Since $i_{v_-v_+}$ is the fixed edge of $\star v_-$, we repeat the gauge-fixing at $v_-$, see Fig.~\ref{fig:gauge-fix-protocol}(c) and proceed along the row.}

{\paragraf{Plaquette step.} At the right-most complete plaquette $p'$ of a row, \eqref{eq:plaquette-move} moves the irrep error of its right edge onto the other three edges of $\partial p'$, where we expand again in the irrep basis \eqref{eq:Emui}. This sets $\mu_n=1$, leaves all multiplication labels unchanged, and can put an irrep error on the shared edge with the next plaquette $p$ to the left, see Fig.~\ref{fig:gauge-fix-protocol}(e). We then repeat this at $p$, see Fig.~\ref{fig:gauge-fix-protocol}(f), which does not change the label fixed at $p'$.}

{\paragraf{Conclusion.} The resulting error operator now has $h_n=\id$ on $\vec i_A$ and $\mu_n=1$ on $\vec i_B$, so $E_{\vec i}^{\vec h,\vec\mu}P$ is a linear combination of the $EP$ with $E\in\cE_{\vec i}$. Since it was constructed from the original error by acting with code stabilizers, it is equivalent to it when acting on the code-space.} On a rough boundary edge, \eqref{eq:boundary-edge-collapse} already gives distinct $G$-irrep labels with identical or zero action. For a proper subgroup $K$, the independent diagonal operators are indexed by $(\Gamma_K,a,b)$, giving $|K|$ labels, as in \eqref{eq:boundary-K-irrep}. 
\end{proof}

We now proceed to prove that the gauge-fixed family \eqref{eq:gauge-fixed-family} obeys the Knill-Laflamme QEC conditions.
\begin{theorem}[Gauge-fixed Knill-Laflamme]\label{thm:KLdegenerate}
For any completed support $\vec i$ that does not host any logical operators, 
any pair $E_{\vec i}^{\vec h,\vec\mu},E_{\vec i}^{\vec k,\vec\nu}\in\cE_{\vec i}$ satisfies
\be\label{eq:KL-degenerate}
P_{c(\vec i)}(E_{\vec i}^{\vec h,\vec\mu})^\dagger E_{\vec i}^{\vec k,\vec\nu}P_{c(\vec i)}
=\delta_{\vec h,\vec k}\delta_{\vec\mu,\vec\nu}P_{c(\vec i)}\,,
\ee
where $P_{c(\vec i)}$ is the local code-space projector \eqref{eq:multi-local-checks} (and the same holds for the global projector $P$). Hence the gauge-fixed errors are orthogonal on the code-space, and their Knill-Laflamme matrix is the identity.
\end{theorem}
\begin{proof}
On each edge, the product $(E_{i_n}^{h_n,\mu_n})^\dagger E_{i_n}^{k_n,\nu_n}$ can be expanded in the single-edge basis \eqref{eq:Emui} and denote the labels by $\vec{g},\vec{\lambda}$. 
First, we consider a term with nontrivial multiplication labels, which are all supported on $\vec i\setminus\vec i_A$. Every dual loop in $\vec i$ is a union of complete stars of $\vec i$, each of which has its fixed edge in $\vec i_A$, so $\vec i\setminus\vec i_A$ contains no dual loop. Since $\vec i$ hosts no logical operator, $\vec i\setminus\vec i_A$ also contains no dual path between two smooth boundary plaquettes, whose end checks would be trivial. Hence there is a plaquette $p\in p(\vec i)$ that contains exactly one of these errored edges. For the full support these are the plaquettes marked by blue crosses in Fig.~\ref{fig:gauge-fix}(a),(b), which contain the last gauge-fixed edge. Lemma~\ref{lemma:PBEPB} then gives $B_pE_{\vec i}^{\vec g,\vec\lambda}B_p=0$, so the term vanishes between the projectors $P_{c(\vec i)}$. The remaining errors are diagonal: those with a nontrivial irrep label are supported on $\vec i\setminus\vec i_B$. Every direct loop in $\vec i$ is a union of complete plaquettes, each with its fixed edge in $\vec i_B$, and $\vec i$ contains no logical and no direct path between two rough boundary vertices. Hence there is a vertex $v\in v(\vec i)$ whose star contains exactly one of the errored edges, with all other edges of $\star v$ without error: these are the vertices marked by blue crosses in Fig.~\ref{fig:gauge-fix}(c),(d). Lemma~\ref{lemma:vertex-violations} then gives $A_vE_{\vec i}^{\vec g,\vec\lambda}A_v=0$. Only the identity term remains, with coefficient $\delta_{\vec h,\vec k}\delta_{\vec\mu,\vec\nu}$ by \eqref{eq:multi-hs}. Since $P_{c(\vec i)}P=P$, the global Knill-Laflamme relation also follows.
\end{proof}

\paragraf{Projector onto Orthogonal Error Sectors.} 
{We can now define the projectors onto the error sectors of the gauge-fixed errors $E_{\vec i}^{\vec h,\vec\mu}\in\cE_{\vec i}$ of \eqref{eq:gauge-fixed-family},
\be\label{eq:error-sector-projector}
P_{c(\vec i)}^{\vec h,\vec\mu}:=E_{\vec i}^{\vec h,\vec\mu}P_{c(\vec i)}(E_{\vec i}^{\vec h,\vec\mu})^\dagger\,.
\ee
These operators are Hermitian and, by Thm.~\ref{thm:KLdegenerate}, they obey 
\be
P_{c(\vec i)}^{\vec h,\vec\mu}P_{c(\vec i)}^{\vec k,\vec\nu}=\delta_{\vec h,\vec k}\delta_{\vec\mu,\vec\nu}P_{c(\vec i)}^{\vec h,\vec\mu}\,,
\ee
so they are mutually orthogonal projectors. The range of $P_{c(\vec i)}^{\vec h,\vec\mu}$ is the image of the local code space under the error $E_{\vec i}^{\vec h,\vec\mu}$. 
We define the \textbf{neutral-cluster projector} as the sum over all the labels of $\cE_{\vec i}$:
\be\label{eq:neutral-cluster-projector}
P_{c(\vec i)}^{\text{neu}}:=\sum_{\vec h,\vec\mu}P_{c(\vec i)}^{\vec h,\vec\mu}\,.
\ee
If $P_{c(\vec i)}^{\text{neu}}\ket\psi=\ket\psi$, then $\ket\psi$ is a code state acted on by gauge-fixed errors supported on $\vec i$, which lie in orthogonal error sectors by Thm.~\ref{thm:KLdegenerate} and  $c(\vec i)$ is a neutral cluster with correctable errors.}

For the following step, we assume a predetermined, neutral-error cluster, such that the support projector \eqref{eq:neutral-cluster-projector} has $+1$ eigenvalue: this identifies the support of each neutral, correctable error cluster.

\section{Error Correction}\label{sec:errorcorrection}

In this section we propose to use a mechanism similar to the deferred measurement, and move the error onto an ancilla, which can then be measured out or discarded. We will first transfer the charge errors onto a vertex ancilla.

\subsection{Charge Error Correction Protocol}\label{sec:charge-transfer}
\paragraf{Setup.} Let $\ket\psi$ be a state invariant under all vertex operators at $v$,
\be
\label{psi-invariant}
A_v^g\ket\psi=\ket\psi\quad\text{for all }g\in G\,.
\ee
Let the edge $i\in\star v$ be acted on by the irrep error $E_i^\mu$ of \eqref{eq:Emui}, with $\mu=(\Gamma,a,b)$, and let $i$ be the only edge of $\star v$ carrying an error. We introduce a $G$-qudit \textbf{vertex ancilla} ($\mathrm{va}$) at $v$, prepared in the state $\ket1_{\mathrm{va}}$ of \eqref{eq:edge-irrep-basis},
\be
\label{ancilla-init}
\ket1_{\mathrm{va}}=\frac{1}{\sqrt{|G|}}\sum_{g\in G}\ket g_{\mathrm{va}}\,,
\ee
and apply the vertex operator controlled by the ancilla,
\be
\label{U-def}
U_v:=\sum_{g\in G}\ketbra g_{\mathrm{va}}\ot A_v^g = \prod_{e\to v}CR_{\mathrm{va},e}\prod_{e\leftarrow v}CL_{\mathrm{va},e}\,,
\ee
where the equality denotes the circuit decomposition in terms of controlled operators, defined as 
\be
\label{controlled-mult}
CL_{c,t}:=\sum_{x\in G}\ketbra x_c\ot L_t^x\,,\quad CR_{c,t}:=\sum_{x\in G}\ketbra x_c\ot R_t^x\,,
\ee

\paragraf{Commutation of $A_v^g$ with the Irrep Error.} The commutation depends on whether $i$ enters or leaves the vertex $v$, in which $A_v^g$ acts on $i$ by $R_i^g$ or $L_i^g$ respectively. As a shorthand, we will denote these cases $i \rightarrow v$ and $i \leftarrow v$, respectively. Using the commutation relations \eqref{eq:ZL_comm2} and \eqref{eq:ZR_comm2}, we obtain 
\be\ba
  R^g_i E_i^{(\Gamma,a,b)} &= \sum_{c=1}^{d_\Gamma}\Gamma(g)_{cb}\,E_i^{(\Gamma,a,c)}R^g_i\\
L^g_i E_i^{(\Gamma,a,b)}& = \sum_{c=1}^{d_\Gamma}\Gamma(g^{-1})_{ac}\,E_i^{(\Gamma,c,b)}L^g_i
\ea\ee

On the other edges of $\star v$ the error acts as the identity,  Therefore
\be
\label{A-E-commutation}
A_v^gE_i^{(\Gamma,a,b)}=
\begin{cases}
\sum_{c=1}^{d_\Gamma}\Gamma(g)_{cb}\,E_i^{(\Gamma,a,c)}A_v^g&i \rightarrow v\,,\\
\sum_{c=1}^{d_\Gamma}\Gamma(g^{-1})_{ac}\,E_i^{(\Gamma,c,b)}A_v^g&i \leftarrow v\,.
\end{cases}
\ee

\paragraf{Action of the Controlled Vertex Operator.} Using \eqref{A-E-commutation} and then \eqref{psi-invariant}, we have
\be
\label{U-on-error}
\ba
\, &
U_v\ket1_{\mathrm{va}}\ot E_i^\mu\ket\psi 
=\frac{1}{\sqrt{|G|}}\sum_{g\in G}\ket g_{\mathrm{va}}\ot A_v^gE_i^{(\Gamma,a,b)}\ket\psi\\
&=\begin{cases}
\sqrt{\frac{d_\Gamma}{|G|}}\sum_{g\in G}\sum_{c=1}^{d_\Gamma}\Gamma(g)_{cb}\ket g_{\mathrm{va}}\ot Z^\Gamma_{ac,i}\ket\psi\,&i \rightarrow v\,,\\
\sqrt{\frac{d_\Gamma}{|G|}}\sum_{g\in G}\sum_{c=1}^{d_\Gamma}\Gamma(g^{-1})_{ac}\ket g_{\mathrm{va}}\ot Z^\Gamma_{cb,i}\ket\psi & i \leftarrow v\,.
\end{cases}
\ea
\ee
Since $\sum_g\Gamma(g)_{cb}\ket g=Z^\Gamma_{cb}\sum_g\ket g=\sqrt{|G|}\,Z^\Gamma_{cb}\ket1$ and $\sum_g\Gamma(g^{-1})_{ac}\ket g=\sqrt{|G|}\,Z^{\Gamma^*}_{ca}\ket1$, this is
\be
\ba
\label{U-on-error-Z}
\, & U_v\ket1_{\mathrm{va}}\ot E_i^\mu\ket\psi\\
&=\begin{cases}
\sqrt{d_\Gamma}\sum_{c=1}^{d_\Gamma}Z^\Gamma_{ac,i}Z^\Gamma_{cb,\mathrm{va}}\bigl(\ket1_{\mathrm{va}}\ot\ket\psi\bigr)\,&i \rightarrow v\,,\\
\sqrt{d_\Gamma}\sum_{c=1}^{d_\Gamma}Z^\Gamma_{cb,i}Z^{\Gamma^*}_{ca,\mathrm{va}}\bigl(\ket1_{\mathrm{va}}\ot\ket\psi\bigr)\,&i \leftarrow v\,,
\end{cases}
\ea
\ee

\paragraf{Controlled Multiplication.} We now apply the controlled left multiplication operator with control $i$ and target $\mathrm{va}$ 

\be\ba
 \label{CL-def}
 CL_{i,\mathrm{va}} &=\sum_{x\in G}\ketbra x_i\ot L_{\mathrm{va}}^x\cr 
 \ea\ee
We use the identities
\be\ba
\label{CL-action}
CL_{i,\mathrm{va}}\Bigl(\sum_{c=1}^{d_\Gamma}Z^\Gamma_{ac,i}Z^\Gamma_{cb,\mathrm{va}}\Bigr)CL_{i,\mathrm{va}}^\dagger=Z^\Gamma_{ab,\mathrm{va}}\,,\\
CL^\dagger_{i,\mathrm{va}}\Bigl(\sum_{c=1}^{d_\Gamma}Z^\Gamma_{cb,i}Z^{\Gamma^*}_{ca,\mathrm{va}}\Bigr)CL_{i,\mathrm{va}}=Z^{\Gamma^*}_{ba,\mathrm{va}}\,.
\ea\ee
which is proven in Appendix ~\ref{app:proof3}. Note that $\sum_{c=1}^{d_\Gamma}Z^\Gamma_{ac,i}Z^\Gamma_{cb,\mathrm{va}}\neq Z^\Gamma_{ab,\mathrm{va}}$ for $\Gamma\neq1$, so \eqref{CL-action} holds only with the conjugation. Moreover, $L^x\ket1=\ket1$ for all $x\in G$, so for every state $\ket\phi$ of the data qudits
\be
\label{CL-vacuum}
CL_{i,\mathrm{va}}\bigl(\ket1_{\mathrm{va}}\ot\ket\phi\bigr)=\ket1_{\mathrm{va}}\ot\ket\phi\,.
\ee
Combining \eqref{U-on-error-Z}, \eqref{CL-action} and \eqref{CL-vacuum}, we obtain
\be
\label{CL-U-result}
\ba
\, &CL_{i,\mathrm{va}}\,U_v\ket1_{\mathrm{va}}\ot E_i^\mu\ket\psi \cr 
&=\sqrt{d_\Gamma}\,Z^\Gamma_{ab,\mathrm{va}}\ket1_{\mathrm{va}}\ot\ket\psi\\
&=\ket{\Gamma_{ab}}_{\mathrm{va}}\ot\ket\psi\,,
\ea
\ee
with $\ket{\Gamma_{ab}}$ the single-qudit state of \eqref{eq:edge-irrep-basis} for the case $i \rightarrow v$ and 
\be
\label{CL-U-result_out}
\ba
\, &CL^\dagger_{i,\mathrm{va}}\,U_v\ket1_{\mathrm{va}}\ot E_i^\mu\ket\psi \cr 
&=\sqrt{d_\Gamma}\,Z^{\Gamma^*}_{ba,\mathrm{va}}\ket1_{\mathrm{va}}\ot\ket\psi\\
&=\ket{\Gamma^*_{ba}}_{\mathrm{va}}\ot\ket\psi\,,
\ea
\ee
for the case $i \leftarrow v$. Thus, define the charge transfer circuit
\be
 T = \begin{cases}
F_\mathrm{va} CL_{i,\mathrm{va}}\,U_v   \,&i \rightarrow v\,, \\
F_\mathrm{va} S_\mathrm{va} CL^\dagger_{i,\mathrm{va}}\,U_v   \,&i \leftarrow v\,, 
\label{charge-transfer-circuit}
 \end{cases}
\ee
where $F$ is the Fourier transform \eqref{eq:fourier-gate} and $S$ is the antipode map which sends $\ket{g} \rightarrow \ket{g^{-1}}$. 
Then
\be
T (\ket1_{\mathrm{va}}\ot E_i^\mu\ket{\psi})  = {\ket{\Gamma_{ab}}_{\mathrm{va}}} \otimes \ket{\psi}\,,
\ee
{where for $i\leftarrow v$ we used $S\ket{\Gamma^*_{ba}}=\ket{\Gamma_{ab}}$, since $\Gamma^*(g^{-1})_{ba}=\Gamma(g)_{ab}$.}

The error is thus moved entirely onto the ancilla, and the data qudits return to $\ket\psi$. Note that since the the ancilla is completely decoupled, and we can in fact discard the ancilla altogether. However, for reusing the ancilla, we can perform a measurement in the group basis, which identifies $\mu$, which can then use to reset the ancilla. This is shown in figure~\ref{fig:charge-transfer-circuit} for a vertex of the square lattice.  

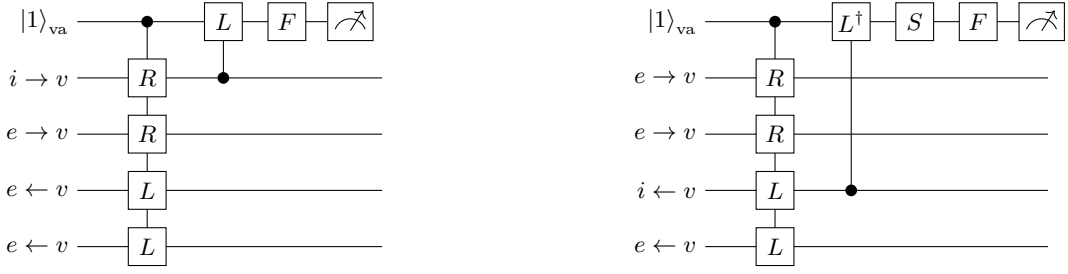
\begin{figure*}
\centering
{\begin{tikzpicture}[x=0.42cm,y=0.62cm,font=\small]
\draw (0,4.8)--(8.4,4.8);
\foreach \y in {0,1.2,2.4,3.6}{\draw (0,\y)--(9.6,\y);}
\node[left] at (0,4.8) {$\ket1_{\mathrm{va}}$};
\node[left] at (0,3.6) {$i\rightarrow v$}; \node[left] at (0,2.4) {$e\to v$};
\node[left] at (0,1.2) {$e\leftarrow v$}; \node[left] at (0,0) {$e\leftarrow v$};
\draw (2.2,4.8)--(2.2,0); \fill (2.2,4.8) circle (2.2pt);
\node[draw,fill=white,minimum size=0.5cm,inner sep=1pt] at (2.2,3.6) {$R$};
\node[draw,fill=white,minimum size=0.5cm,inner sep=1pt] at (2.2,2.4) {$R$};
\node[draw,fill=white,minimum size=0.5cm,inner sep=1pt] at (2.2,1.2) {$L$};
\node[draw,fill=white,minimum size=0.5cm,inner sep=1pt] at (2.2,0) {$L$};
\draw (4.6,3.6)--(4.6,4.8); \fill (4.6,3.6) circle (2.2pt);
\node[draw,fill=white,minimum size=0.5cm,inner sep=1pt] at (4.6,4.8) {$L$};
\node[draw,fill=white,minimum size=0.5cm,inner sep=1pt] at (6.6,4.8) {$F$};
\node[draw,fill=white,minimum width=0.6cm,minimum height=0.5cm] (mc) at (8.6,4.8) {};
\draw (mc.center) ++(-0.2cm,-0.1cm) arc[start angle=160,end angle=20,radius=0.21cm];
\draw[->] (mc.center) ++(0,-0.12cm) -- ++(0.17cm,0.27cm);
\end{tikzpicture}} \hspace{3 cm}
{\begin{tikzpicture}[x=0.42cm,y=0.62cm,font=\small]
\draw (0,4.8)--(10.4,4.8);
\foreach \y in {0,1.2,2.4,3.6}{\draw (0,\y)--(10.8,\y);}
\node[left] at (0,4.8) {$\ket1_{\mathrm{va}}$};
\node[left] at (0,3.6) {$e\rightarrow v$}; \node[left] at (0,2.4) {$e\to v$};
\node[left] at (0,1.2) {$i\leftarrow v$}; \node[left] at (0,0) {$e\leftarrow v$};
\draw (2.2,4.8)--(2.2,0); \fill (2.2,4.8) circle (2.2pt);
\node[draw,fill=white,minimum size=0.5cm,inner sep=1pt] at (2.2,3.6) {$R$};
\node[draw,fill=white,minimum size=0.5cm,inner sep=1pt] at (2.2,2.4) {$R$};
\node[draw,fill=white,minimum size=0.5cm,inner sep=1pt] at (2.2,1.2) {$L$};
\node[draw,fill=white,minimum size=0.5cm,inner sep=1pt] at (2.2,0) {$L$};
\draw (4.6,1.2)--(4.6,4.8); \fill (4.6,1.2) circle (2.2pt);
\node[draw,fill=white,minimum size=0.5cm,inner sep=1pt] at (4.6,4.8) {$L^\dagger$};
\node[draw,fill=white,minimum size=0.5cm,inner sep=1pt] at (6.6,4.8) {$S$};
\node[draw,fill=white,minimum size=0.5cm,inner sep=1pt] at (8.6,4.8) {$F$};
\node[draw,fill=white,minimum width=0.6cm,minimum height=0.5cm] (mc) at (10.6,4.8) {};
\draw (mc.center) ++(-0.2cm,-0.1cm) arc[start angle=160,end angle=20,radius=0.21cm];
\draw[->] (mc.center) ++(0,-0.12cm) -- ++(0.17cm,0.27cm);
\end{tikzpicture}}
\caption{{The circuit \eqref{charge-transfer-circuit} for the charge transfer at a vertex $v$ of the square lattice, with the irrep error $E_i^\mu$ for an edge $i$ entering $v$ (left) and for an edge $i$ exiting $v$ (right). A dot on a control qudit in the group basis state $\ket x$, connected to boxes $L$ or $R$, applies $L^x$ or $R^x$ to the target qudits, as in \eqref{controlled-mult}. The first gate is the controlled vertex operator $U_v$ of \eqref{U-def}, with the ancilla as control, and the second is $CL_{i,\mathrm{va}}$ or $CL_{i,\mathrm{va}}^\dagger$ (depending on the orientation of $i$), with $i$ as control. Afterwards the data qudits are in $\ket\psi$ and the ancilla is in $\ket{\Gamma_{ab}}$, which the Fourier transform $F$ maps to the group basis state $\ket{\Gamma,a,b}$ before the measurement in the group basis.}}
\label{fig:charge-transfer-circuit}
\end{figure*}

\paragraf{Example: $G=\Z_2$.} Let $\Z_2=\{0,1\}$, with group basis $\ket0,\ket1$ and $\ket\pm:=(\ket0\pm\ket1)/\sqrt2$. In this example $\ket+$ is the state $\ket1$ of \eqref{eq:edge-irrep-basis}, and for the sign irrep $\ket{\Gamma_{11}}=\ket-$. Moreover $L^1=R^1=X$, $E_i^\mu=Z_i$, and $A_v^1=\prod_{e\in\star v}X_e$, so we do not have to distinguish the two cases. The controlled vertex operator \eqref{U-def} is
\be
\label{U-Z2}
U_v=\ketbra0_{\mathrm{va}}\ot\bbI+\ketbra1_{\mathrm{va}}\ot A_v^1\,.
\ee
Since $i\in\star v$, the operators $A_v^1$ and $Z_i$ anticommute. Together with $A_v^1\ket\psi=\ket\psi$ this gives
\be
\label{U-Z2-on-error}
\ba
\, &U_v\ket+_{\mathrm{va}}\ot Z_i\ket\psi\cr 
&=\frac{1}{\sqrt2}\Bigl(\ket0_{\mathrm{va}}\ot Z_i\ket\psi+\ket1_{\mathrm{va}}\ot A_v^1Z_i\ket\psi\Bigr)\\
&=\frac{1}{\sqrt2}\Bigl(\ket0_{\mathrm{va}}-\ket1_{\mathrm{va}}\Bigr)\ot Z_i\ket\psi\\
&=\ket-_{\mathrm{va}}\ot Z_i\ket\psi\\
&=Z_{\mathrm{va}}Z_i\,\ket+_{\mathrm{va}}\ot\ket\psi\,.
\ea
\ee
The controlled multiplication \eqref{CL-def} is $CL_{i,\mathrm{va}}=\mathrm{CNOT}_{i,\mathrm{va}}$, with control $i$ and target $\mathrm{va}$. It satisfies
\be
\label{CNOT-Z2}
\ba
\mathrm{CNOT}_{i,\mathrm{va}}\,Z_{\mathrm{va}}Z_i\,\mathrm{CNOT}_{i,\mathrm{va}}^\dagger&=Z_{\mathrm{va}}\,,\\
\mathrm{CNOT}_{i,\mathrm{va}}\bigl(\ket+_{\mathrm{va}}\ot\ket\phi\bigr)&=\ket+_{\mathrm{va}}\ot\ket\phi\,,
\ea
\ee
which are \eqref{CL-action} and \eqref{CL-vacuum} for $\Z_2$. Therefore
\be\ba 
\label{CNOT-Z2-result}
\, &\mathrm{CNOT}_{i,\mathrm{va}}\,U_v\ket+_{\mathrm{va}}\ot Z_i\ket\psi \cr 
&=Z_{\mathrm{va}}\ket+_{\mathrm{va}}\ot\ket\psi=\ket-_{\mathrm{va}}\ot\ket\psi\,,
\ea
\ee
in agreement with \eqref{CL-U-result}.

\paragraf{Example: $G=S_3$.} Let $S_3=\langle r,s\,|\,r^3=s^2=\id,\ srs=r^{-1}\rangle$, with elements $r^k$ and $sr^k$ for $k=0,1,2$. Its irreps are the trivial and the sign irrep, both one-dimensional, and the two-dimensional irrep $\Gamma$ with
\be
\label{S3-irrep}
\Gamma(r^k)=\begin{pmatrix}\omega^k&0\\0&\omega^{-k}\end{pmatrix},\qquad
\Gamma(sr^k)=\begin{pmatrix}0&\omega^{-k}\\\omega^k&0\end{pmatrix}\,,
\ee
where $\omega:=e^{2\pi i/3}$.
Here, we consider the case $i \rightarrow v$ without loss of generality. For the one-dimensional irreps the sum over $c$ in \eqref{U-on-error-Z} has a single term, and the computation is as for $\Z_2$. The new feature appears for $\Gamma$, and we consider $\mu=(\Gamma,1,2)$. The single-qudit states \eqref{eq:edge-irrep-basis} that enter are
\be
\label{S3-states}
\ba
\ket{\Gamma_{12}}&=\frac{1}{\sqrt3}\bigl(\ket s+\omega^2\ket{sr}+\omega\ket{sr^2}\bigr)\,,\\
\ket{\Gamma_{22}}&=\frac{1}{\sqrt3}\bigl(\ket\id+\omega^2\ket r+\omega\ket{r^2}\bigr)\,.
\ea
\ee
Since $\sqrt2\,Z^\Gamma_{cb}\ket1=\ket{\Gamma_{cb}}$, \eqref{U-on-error-Z} gives
\be
\label{S3-U-on-error}
\ba
U_v\ket1_{\mathrm{va}}\ot E_i^{(\Gamma,1,2)}\ket\psi
&=\ket{\Gamma_{12}}_{\mathrm{va}}\ot Z^\Gamma_{11,i}\ket\psi\\
&\quad+\ket{\Gamma_{22}}_{\mathrm{va}}\ot Z^\Gamma_{12,i}\ket\psi\,.
\ea
\ee
Unlike for $\Z_2$, the ancilla is entangled with the data qudits, and the row index $a=1$ of the error is still on the data qudits. A Fourier measurement of the ancilla at this stage would return $(\Gamma,1,2)$ or $(\Gamma,2,2)$, each with probability $\frac12$ since the label of $i$ is uniformly distributed, and leave the data qudits in $Z^\Gamma_{11,i}\ket\psi$ or $Z^\Gamma_{12,i}\ket\psi$, up to normalization. The controlled multiplication removes this entanglement. From $\Gamma(x^{-1}g)=\Gamma(x^{-1})\Gamma(g)$,
\be
\label{L-on-irrep-states}
L^x\ket{\Gamma_{cb}}=\sum_{c'=1}^{d_\Gamma}\Gamma(x^{-1})_{cc'}\ket{\Gamma_{c'b}}\,.
\ee
If $i$ carries the label $x=r^k$, only the first term of \eqref{S3-U-on-error} contributes, with $\Gamma(r^k)_{11}=\omega^k$, and $L^{r^k}\ket{\Gamma_{12}}=\omega^{-k}\ket{\Gamma_{12}}$. If $i$ carries $x=sr^k$, only the second term contributes, with $\Gamma(sr^k)_{12}=\omega^{-k}$, and $L^{sr^k}\ket{\Gamma_{22}}=\omega^k\ket{\Gamma_{12}}$, since $(sr^k)^{-1}=sr^k$. In both cases the phases cancel,
\be
\label{S3-CL}
\ba
&CL_{i,\mathrm{va}}\Bigl(\ket{\Gamma_{12}}_{\mathrm{va}}\ot Z^\Gamma_{11,i}+\ket{\Gamma_{22}}_{\mathrm{va}}\ot Z^\Gamma_{12,i}\Bigr)\ket x_i\\
&=\ket{\Gamma_{12}}_{\mathrm{va}}\ot\ket x_i\quad\text{for all }x\in S_3\,,
\ea
\ee
so that $CL_{i,\mathrm{va}}U_v\ket1_{\mathrm{va}}\ot E_i^{(\Gamma,1,2)}\ket\psi=\ket{\Gamma_{12}}_{\mathrm{va}}\ot\ket\psi$, in agreement with \eqref{CL-U-result}.

\subsection{Flux Error Correction Protocol}\label{sec:flux-transfer}
We now develop the analogous protocol for the plaquette terms, which moves a multiplication error onto an ancilla.

\paragraf{Setup.} Let $\ket\psi$ be a state with trivial flux through a plaquette $p$,
\be
\label{psi-flat}
B_p\ket\psi=\ket\psi\,,
\ee
with $B_p=B_p^{\id}$ of \eqref{eq:Av_Bp}. We consider a multiplication error $L_i^h$ of \eqref{eq:Emui} acting on the edge $i\in\partial p$ and we take $i$ to be the only edge of $\de p$ carrying a multiplication error: as explained in the proof of Thm.~\ref{thm:KLdegenerate}, in each error cluster there exists an edge with this property after suitable gauge-fixing. In the flux correction protocol we now describe, we will always take the basepoint of the plaquette holonomy to be the source vertex $v_-$ of the edge $i$, and follow the path around $p$ starting along $i$ in the direction of the edge orientation 
e.g. in \eqref{eq:AvBp} the basepoint $v_-$ is the bottom-left vertex  and we follow the edges in the order $e_1,\ldots,e_4$, which carry the group labels $g_1,\ldots,g_4$. Then $B_p^g$ projects onto the states with holonomy $g_1g_2g_3^{-1}g_4^{-1}=g$ based at $v_-$. For an error on any other edge of $p$, the holonomy based at its source vertex follows by cyclically relabeling the edges, and reversing the direction of the loop if needed. Changing the basepoint can conjugate the holonomy, but the action of $B_p=B_p^{\id}$ does not depend on the basepoint choice. 
In the following, we will use $B_p^g$ to compute the holonomy with the basepoint and orientation that depends on the edge $i$ we are correcting: the basepoint is $v_-(i)$ and the orientation is the one along edge $i$. For each plaquette $p$ and choice of basepoint and orientation, the $B_p^g$ operators are orthogonal projectors with $\sum_gB_p^g=\bbI$, so \eqref{psi-flat} is equivalent to
\be
\label{psi-flat-g}
B_p^g\ket\psi=\delta_{g,\id}\ket\psi\,.
\ee
We introduce a $G$-qudit \textbf{plaquette ancilla} ($\mathrm{pa}$) at $p$, prepared in the group basis state $\ket\id_{\mathrm{pa}}$, and apply the right multiplication controlled by the holonomy of $p$ via $B_p^g$ and 
acting on the plaquette ancilla (which we write in the first tensor-product)
\be
\label{V-def}
V_p:=\sum_{g\in G}R_{\mathrm{pa}}^{g^{-1}}\ot B_p^g\,.
\ee
It is unitary, since the $B_p^g$ are orthogonal projectors summing to $\bbI$, and on a state with holonomy $g$ it maps the ancilla state $\ket y_{\mathrm{pa}}$ to $\ket{yg}_{\mathrm{pa}}$.

We can commute the $B_p^g$ through the multiplication error as follows: since $L_i^h$ maps $g_i$ to $hg_i$, it maps a state with holonomy $k$ to a state with holonomy $hk$. Therefore
\be
\label{B-L-commutation}
B_p^gL_i^h=L_i^hB_p^{h^{-1}g}\,.
\ee

\paragraf{Action of the Controlled Holonomy.} Using $R^{g^{-1}}\ket\id=\ket g$, \eqref{B-L-commutation} and \eqref{psi-flat-g}, we compute
\be
\label{V-on-error}
\ba
V_p\ket\id_{\mathrm{pa}}\ot L_i^h\ket\psi
&=\sum_{g\in G}\ket g_{\mathrm{pa}}\ot B_p^gL_i^h\ket\psi\\
&=\sum_{g\in G}\ket g_{\mathrm{pa}}\ot L_i^hB_p^{h^{-1}g}\ket\psi\\
&=\ket h_{\mathrm{pa}}\ot L_i^h\ket\psi\,.
\ea
\ee
Unlike \eqref{U-on-error-Z}, the ancilla and the data qudits are now already in a product state.

\paragraf{Error Transfer to the Ancilla.} To remove the error from the state $\ket{\psi}$, we apply the inverse of the controlled left multiplication \eqref{controlled-mult} with control $\mathrm{pa}$ and target $i$,
\be
\label{CL-dagger}
CL_{\mathrm{pa},i}^\dagger=\sum_{h\in G}\ketbra h_{\mathrm{pa}}\ot L_i^{h^{-1}}\,,
\ee
which acts as follows:
\be
\label{flux-transfer-result}
CL_{\mathrm{pa},i}^\dagger\,V_p(\ket\id_{\mathrm{pa}}\ot L_i^h\ket\psi)=\ket h_{\mathrm{pa}}\ot\ket\psi\,.
\ee
The error is therefore moved entirely onto the ancilla, and a measurement of the ancilla in the group basis identifies $h$. Since $E_i^\mu$ is diagonal in the group basis, it commutes with every $B_p^g$, and the same steps for the error $E_i^{h,\mu}=L_i^hE_i^\mu$ of \eqref{eq:Emui} give
\be
\label{flux-transfer-full}
CL_{\mathrm{pa},i}^\dagger\,V_p\ket\id_{\mathrm{pa}}\ot E_i^{h,\mu}\ket\psi=\ket h_{\mathrm{pa}}\ot E_i^\mu\ket\psi\,.
\ee
If $\ket\psi$ also satisfies \eqref{psi-invariant}, the remaining irrep error $E_i^\mu$ is moved onto a second ancilla by the charge transfer of Sec.~\ref{sec:charge-transfer}.

\paragraf{Circuit.} The controlled holonomy is a product of controlled right multiplications of the ancilla, controlled by the edges of $p$. The operator $CR_{e_1,\mathrm{pa}}^\dagger$ acting on $\ket y_{\mathrm{pa}}$ maps $y$ to $yg_1$, then $CR_{e_2,\mathrm{pa}}^\dagger$ to $yg_1g_2$, then $CR_{e_3,\mathrm{pa}}$ and $CR_{e_4,\mathrm{pa}}$ multiply from the right by $g_3^{-1}$ and $g_4^{-1}$. Hence
\be
\label{V-gates}
V_p=CR_{e_4,\mathrm{pa}}\,CR_{e_3,\mathrm{pa}}\,CR_{e_2,\mathrm{pa}}^\dagger\,CR_{e_1,\mathrm{pa}}^\dagger\,,
\ee
where, compared to \eqref{U-def}, the roles of control and target are exchanged. For an edge $i$ in another position on $\partial p$, $V_p$ is the product of $CR_{e,\mathrm{pa}}^\dagger$ for the edges $e$ traversed along their orientation and $CR_{e,\mathrm{pa}}$ for the edges traversed against it, applied in the order of the holonomy based at the source vertex $v_-$ of $i$. The flux error correction is therefore the circuit
\be
\label{flux-transfer-circuit}
CL_{\mathrm{pa},i}^\dagger\,CR_{e_4,\mathrm{pa}}\,CR_{e_3,\mathrm{pa}}\,CR_{e_2,\mathrm{pa}}^\dagger\,CR_{e_1,\mathrm{pa}}^\dagger\,,
\ee
acting on $\ket\id_{\mathrm{pa}}\ot L_i^h\ket\psi$ and followed by a measurement of the ancilla in the group basis, shown in figure~\ref{fig:flux-transfer-circuit}. 

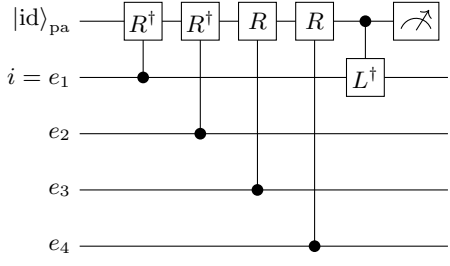
\begin{figure}
\centering
{\begin{tikzpicture}[x=0.42cm,y=0.62cm,font=\small]
\draw (0,4.8)--(10.4,4.8);
\foreach \y in {0,1.2,2.4,3.6}{\draw (0,\y)--(11.6,\y);}
\node[left] at (0,4.8) {$\ket\id_{\mathrm{pa}}$};
\node[left] at (0,3.6) {$i=e_1$}; \node[left] at (0,2.4) {$e_2$};
\node[left] at (0,1.2) {$e_3$}; \node[left] at (0,0) {$e_4$};
\draw (2,3.6)--(2,4.8); \fill (2,3.6) circle (2.2pt);
\node[draw,fill=white,minimum size=0.5cm,inner sep=1pt] at (2,4.8) {$R^\dagger$};
\draw (3.8,2.4)--(3.8,4.8); \fill (3.8,2.4) circle (2.2pt);
\node[draw,fill=white,minimum size=0.5cm,inner sep=1pt] at (3.8,4.8) {$R^\dagger$};
\draw (5.6,1.2)--(5.6,4.8); \fill (5.6,1.2) circle (2.2pt);
\node[draw,fill=white,minimum size=0.5cm,inner sep=1pt] at (5.6,4.8) {$R$};
\draw (7.4,0)--(7.4,4.8); \fill (7.4,0) circle (2.2pt);
\node[draw,fill=white,minimum size=0.5cm,inner sep=1pt] at (7.4,4.8) {$R$};
\draw (9,4.8)--(9,3.6); \fill (9,4.8) circle (2.2pt);
\node[draw,fill=white,minimum size=0.5cm,inner sep=1pt] at (9,3.6) {$L^\dagger$};
\node[draw,fill=white,minimum width=0.6cm,minimum height=0.5cm] (mf) at (10.6,4.8) {};
\draw (mf.center) ++(-0.2cm,-0.1cm) arc[start angle=160,end angle=20,radius=0.21cm];
\draw[->] (mf.center) ++(0,-0.12cm) -- ++(0.17cm,0.27cm);
\end{tikzpicture}}
\caption{{The circuit \eqref{flux-transfer-circuit} for the flux transfer at a plaquette $p$ of the square lattice, with the edges labeled as in \eqref{eq:AvBp} and the multiplication error $L_i^h$ on $i=e_1$. A dot on a control qudit in the group basis state $\ket x$, connected to a box $R$, $R^\dagger$ or $L^\dagger$, applies $R^x$, $R^{x^{-1}}$ or $L^{x^{-1}}$ to the target qudit, as in \eqref{controlled-mult}. The first four gates are the controlled holonomy $V_p$ of \eqref{V-gates}, after which the ancilla is in the state $\ket h$, and the last gate $CL_{\mathrm{pa},i}^\dagger$ removes the error from $i$. The ancilla is then measured in the group basis.}}
\label{fig:flux-transfer-circuit}
\end{figure}

\paragraf{Example: $G=\Z_2$.} For $\Z_2=\{0,1\}$, $L^1=R^1=X$, the multiplication error is $X_i$, the plaquette ancilla starts in $\ket\id=\ket0$, and every controlled multiplication \eqref{controlled-mult} is a CNOT. The controlled holonomy operator \eqref{V-gates} is 
\be
V_p=\prod_{k=1}^4\mathrm{CNOT}_{e_k,\mathrm{pa}}\,,
\ee
which writes the parity of the labels of $p$ onto the ancilla. Since a code-state $\ket\psi$ has even parity and $X_i$ flips it, 
\be
\label{flux-Z2}
\ba
&\mathrm{CNOT}_{\mathrm{pa},i}\prod_{k=1}^4\mathrm{CNOT}_{e_k,\mathrm{pa}} \Big( \ket0_{\mathrm{pa}}\ot X_i\ket\psi\Big) \\
&=\mathrm{CNOT}_{\mathrm{pa},i}\Big( \ket1_{\mathrm{pa}}\ot X_i\ket\psi \Big)
=\ket1_{\mathrm{pa}}\ot\ket\psi\,,
\ea
\ee
in agreement with \eqref{flux-transfer-result} and the standard $\Z_2$ syndrome extraction circuits.

\paragraf{Example: $G=S_3$.} We use the presentation of $S_3$ of Sec.~\ref{sec:charge-transfer} and consider the error $L^s$ on each of the edges $e_1,\ldots,e_4$ of $p$ in turn, with the edges oriented as in \eqref{eq:AvBp}. For the error $L_{e_k}^s$ the basepoint of the plaquette is the source vertex of $e_k$, which is the bottom-left, top-left, bottom-right and bottom-left corner of $p$ for $k=1,2,3,4$, so the holonomy starting along $e_k$ is:
\be\ba\label{S3-basepoints}
e_1&:\ g_1g_2g_3^{-1}g_4^{-1}\,,\qquad e_2:\ g_2g_3^{-1}g_4^{-1}g_1\,,\cr
e_3&:\ g_3g_2^{-1}g_1^{-1}g_4\,,\qquad e_4:\ g_4g_3g_2^{-1}g_1^{-1}\,.
\ea\ee
For $e_1$ and $e_2$ the loop runs clockwise, while for $e_3$ and $e_4$ it runs anticlockwise, so that $e_k$ is always traversed along its orientation. Each of these holonomies is a cyclic permutation of $g_1g_2g_3^{-1}g_4^{-1}$ or of its inverse, so it is trivial on any state with trivial flux $g_1g_2g_3^{-1}g_4^{-1}=\id$. The error replaces $g_k$ by $sg_k$, which, with the basepoint and orientation conventions described above, multiplies the holonomy by $s$ from the left. For example, for an error on $e_1$ the four gates of \eqref{V-gates} act on the ancilla as
\be
\label{S3-flux-gates}
\ba
\ket\id&\mapsto\ket{sg_1}\mapsto\ket{sg_1g_2}\\
&\mapsto\ket{sg_1g_2g_3^{-1}}\mapsto\ket{sg_1g_2g_3^{-1}g_4^{-1}}=\ket s\,,
\ea
\ee
while for $e_3$ the gates $CR_{e_3,\mathrm{pa}}^\dagger$, $CR_{e_2,\mathrm{pa}}$, $CR_{e_1,\mathrm{pa}}$ and $CR_{e_4,\mathrm{pa}}^\dagger$, applied in this order, give
\be
\label{S3-flux-gates-e3}
\ba
\ket\id&\mapsto\ket{sg_3}\mapsto\ket{sg_3g_2^{-1}}\\
&\mapsto\ket{sg_3g_2^{-1}g_1^{-1}}\mapsto\ket{sg_3g_2^{-1}g_1^{-1}g_4}=\ket s\,,
\ea
\ee
and similarly for $e_2$ and $e_4$. In all four cases the ancilla ends in $\ket s$ independently of the labels, in agreement with \eqref{V-on-error}, and the gate $CL_{\mathrm{pa},e_k}^\dagger$ applies $L^{s^{-1}}=L^s$ to $e_k$, which removes the error. The ancilla then records the group element $s$.

\subsection{Correction of Gauge-Fixed Supports}\label{sec:gauge-fixed-correction}

We now combine the charge and flux error transfer circuits we just computed, in order to correct the gauge-fixed errors 
\eqref{eq:gauge-fixed-family} on the supports of Fig.~\ref{fig:gauge-fix}. 
Each error transfer circuit moves a single label of the error onto its own ancilla. 
We use the right gauge-fixing as in Fig.~\ref{fig:gauge-fix} and start at the blue crosses, which show the checks whose support contains a single edge error. We then follow the support from right to left. 
Let us first summarize the flux and charge corrections and then the protocol to correct a general error:

\paragraf{Flux Correction.} The flux correction circuit of Sec.~\ref{sec:flux-transfer} at a plaquette $p$ removes the multiplication error on an edge $i\in\partial p$, under the assumption that $i$ is the only edge of $\partial p$ carrying an uncorrected multiplication error: 
this holds for the plaquettes shown with blue crosses in Fig.~\ref{fig:gauge-fix} (a) and (b). The irrep errors on $\partial p$ are diagonal in the group basis and commute with every $B_p^g$ and the errors on edges outside $\partial p$ commute with $V_p$ and $CL_{\mathrm{pa},i}^\dagger$, so \eqref{flux-transfer-full} applies. 

\paragraf{Charge Correction.}
The charge correction circuit of Sec.~\ref{sec:charge-transfer} at a vertex $v$ removes the irrep error on an edge $i\in\star v$, provided that $i$ is the only edge of $\star v$ carrying an uncorrected irrep error and that no edge of $\star v$ carries a multiplication error: this holds for the vertices marked by blue crosses in Fig.~\ref{fig:gauge-fix} (c)-(d). The remaining errors then commute with $A_v^g$, and \eqref{CL-U-result} applies. 

\paragraf{General Protocol.} The correction protocol proceeds as follows: 
we first remove all multiplication errors using the flux transfer correction, moving along the path of the dual lattice between the marked plaquettes of Fig.~\ref{fig:gauge-fix}. 
The flux correction for each edge of $i$ refers to the basepoint $v_-$ of $i$, as in Sec.~\ref{sec:flux-transfer}. 
After this, all multiplication errors have been corrected and the data qudits are in the state $E_{\vec i}^{\vec\mu}\ket\psi$.
Subsequently, we remove all irrep errors in a charge correction step by 
following Sec.~\ref{sec:charge-transfer}. 
For the flux error correction, the edges with uncorrected multiplication errors never contain a dual loop. Therefore, as in the proof of Theorem~\ref{thm:KLdegenerate}, there is always a plaquette of $p(\vec i)$ whose boundary contains exactly one of them. 
In the charge correction, the edges with uncorrected irrep errors never contain a direct loop, so there is always a vertex of $v(\vec i)$ whose star contains exactly one of them.

\paragraf{Notation.} We write $v_\pm(i)$ for the endpoints of $i$, with $i$ oriented from $v_-(i)$ to $v_+(i)$, and $p_\nearrow(v)$, $p_\searrow(v)$, $p_\swarrow(v)$ and $p_\nwarrow(v)$ for the four plaquettes with corner $v$. 

\paragraf{One Star, Fig.~\ref{fig:gauge-fix}(a).} 
We consider the one-star error configuration, where we provide a few more details compared with Fig.~\ref{fig:gauge-fix}:
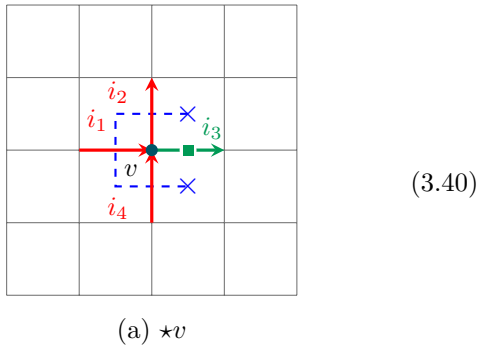
\begin{equation}
    \begin{split}
      \begin{tikzpicture}[scale=1.6,>=stealth]
        \draw[step=0.6,black!55] (0,-0.6) grid (2.4,1.8);
        \draw[blue,dashed,thick] (1.5,0.9) -- (0.9,0.9) -- (0.9,0.3) -- (1.5,0.3);
        \draw[red,very thick,->] (0.6,0.6) -- (1.2,0.6);
        \draw[red,very thick,->] (1.2,0.6) -- (1.2,1.2);
        \draw[red,very thick,->] (1.2,0) -- (1.2,0.6);
        \draw[ForestGreen,very thick,->] (1.2,0.6) -- (1.8,0.6);
        \fill[blue(munsell)] (1.2,0.6) circle (0.05);
        \fill[white] (1.5,0.6) circle (0.07);
        \fill[ForestGreen] (1.5,0.6) +(-0.045,-0.045) rectangle +(0.045,0.045);
        \node[blue,font=\large] at (1.5,0.9) {$\times$};
        \node[blue,font=\large] at (1.5,0.3) {$\times$};
        \node[red] at (0.75,0.85) {$i_1$};
        \node[red] at (0.92,1.08) {$i_2$};
        \node[ForestGreen] at (1.7,0.78) {$i_3$};
        \node[red] at (0.92,0.12) {$i_4$};
        \node at (1.03,0.43) {$v$};
        \node at (1.2,-0.9) {(a) $\star v$};
      \end{tikzpicture}
    \end{split}
  \end{equation}
The edges $i_1,i_4$ are oriented towards $v$, while $i_2,i_3$ are oriented away from $v$. The gauge-fixed edge $i_3$ carries $h_3=\id$. The flux correction stage is:
\begin{enumerate}
\item At the blue crosses $p_\nearrow(v)$ and $p_\searrow(v)$, the only multiplication errors are on $i_2$ and $i_4$, since $h_3=\id$. They are corrected using the flux correction protocol of Sec.~\ref{sec:flux-transfer} with the plaquette basepoints $v_-(i_2)$ and $v_-(i_4)$ respectively.
\item Moving left, $i_1$ is now the only edge of the cluster carrying a multiplication error. It is corrected taking the plaquette $p_\nwarrow(v)$ with the basepoint $v_-(i_1)$ and counter-clockwise orientation.
\end{enumerate}
For the charge correction, we use the protocol of Sec.~\ref{sec:charge-transfer} in parallel for the four edges with vertex ancillas on the non-shared endpoints.

\paragraf{Two Stars, Fig.~\ref{fig:gauge-fix}(b).} 
The two-star configuration is 
 \be
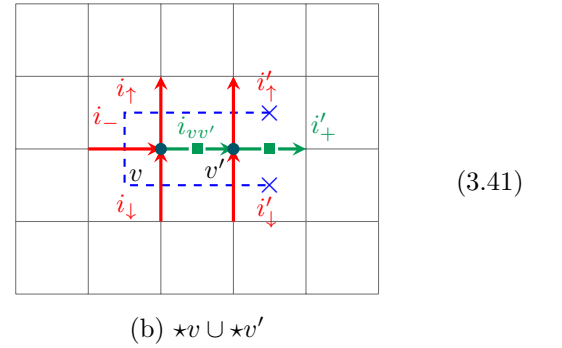
 \begin{split}
  \begin{tikzpicture}[scale=1.6,>=stealth]
    \draw[step=0.6,black!55] (-0.6,-0.6) grid (2.4,1.8);
    \draw[blue,dashed,thick] (1.5,0.9) -- (0.3,0.9) -- (0.3,0.3) -- (1.5,0.3);
    \draw[red,very thick,->] (0,0.6) -- (0.6,0.6);
    \draw[red,very thick,->] (0.6,0.6) -- (0.6,1.2);
    \draw[red,very thick,->] (0.6,0) -- (0.6,0.6);
    \draw[red,very thick,->] (1.2,0.6) -- (1.2,1.2);
    \draw[red,very thick,->] (1.2,0) -- (1.2,0.6);
    \draw[ForestGreen,very thick,->] (0.6,0.6) -- (1.2,0.6);
    \draw[ForestGreen,very thick,->] (1.2,0.6) -- (1.8,0.6);
    \fill[blue(munsell)] (0.6,0.6) circle (0.05);
    \fill[blue(munsell)] (1.2,0.6) circle (0.05);
    \fill[white] (0.9,0.6) circle (0.07);
    \fill[ForestGreen] (0.9,0.6) +(-0.045,-0.045) rectangle +(0.045,0.045);
    \fill[white] (1.5,0.6) circle (0.07);
    \fill[ForestGreen] (1.5,0.6) +(-0.045,-0.045) rectangle +(0.045,0.045);
    \node[blue,font=\large] at (1.5,0.9) {$\times$};
    \node[blue,font=\large] at (1.5,0.3) {$\times$};
    \node[red] at (0.15,0.85) {$i_-$};
    \node[red] at (0.32,1.08) {$i_\uparrow$};
    \node[red] at (0.32,0.12) {$i_\downarrow$};
    \node[red] at (1.48,1.12) {$i'_\uparrow$};
    \node[red] at (1.48,0.08) {$i'_\downarrow$};
    \node[ForestGreen] at (0.9,0.78) {$i_{vv'}$};
    \node[ForestGreen] at (1.95,0.78) {$i'_+$};
    \node at (0.43,0.44) {$v^{\phantom{'}}$};
    \node at (1.05,0.44) {$v'$};
    \node at (0.9,-0.9) {(b) $\star v\cup\star v'$};
  \end{tikzpicture}
  \end{split}
  \ee
The gauge-fixed edges $i_{vv'}$ and $i'_+$ carry trivial multiplication labels. The flux correction stage follows the path on the dual lattice between the marked plaquettes is the following:
\begin{enumerate}
\item At the marked plaquettes $p_\nearrow(v')$ and $p_\searrow(v')$, transfer the multiplication errors of $i'_\uparrow$ and $i'_\downarrow$ to ancillas and remove them.
\item Moving left, $i_\uparrow$ and $i_\downarrow$ are now the only edges carrying multiplication errors on $\partial p_\nearrow(v)$ and $\partial p_\searrow(v)$ so we again transfer to ancillas and remove them.
\item Moving left again, we transfer the multiplication error of $i_-$ at $p_\nwarrow(v)$, with the basepoint $v_-(i_-)$.
\end{enumerate}
The charge correction requires two rounds:
\begin{enumerate}
    \item We proceed in parallel for all edges other than $i_{vv'}$, since each of them has an endpoint other than $v,v'$ whose star contains no other errors.
    \item Next, we correct transfer the charge on $i_{vv'}$, transferring it to an ancilla at $v'$, since $i_{vv'}$ is now the only edge of $\star v'$ carrying an error. 
\end{enumerate} 

\paragraf{One Plaquette, Fig.~\ref{fig:gauge-fix}(c).} 
The single plaquette configuration is 
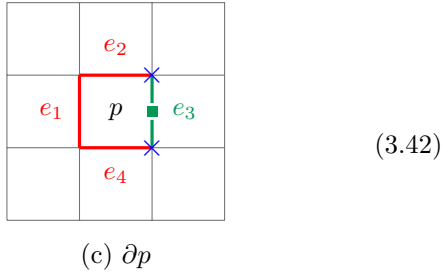
 \begin{equation}
    \begin{split}
      \begin{tikzpicture}[scale=1.6]
        \draw[step=0.6,black!55] (0.6,0) grid (2.4,1.8);
        \draw[red,very thick] (1.2,1.2) -- (1.8,1.2);
        \draw[ForestGreen,very thick] (1.8,0.6) -- (1.8,1.2);
        \draw[red,very thick] (1.2,0.6) -- (1.8,0.6);
        \draw[red,very thick] (1.2,0.6) -- (1.2,1.2);
        \fill[white] (1.8,0.9) circle (0.07);
        \fill[ForestGreen] (1.8,0.9) +(-0.045,-0.045) rectangle +(0.045,0.045);
        \node[blue,font=\large] at (1.8,1.2) {$\times$};
        \node[blue,font=\large] at (1.8,0.6) {$\times$};
        \node[red] at (0.97,0.9) {$e_1$};
        \node[red] at (1.5,1.43) {$e_2$};
        \node[ForestGreen] at (2.07,0.9) {$e_3$};
        \node[red] at (1.5,0.37) {$e_4$};
        \node at (1.5,0.9) {$p$};
        \node at (1.5,-0.3) {(c) $\partial p$};
      \end{tikzpicture}
    \end{split}
  \end{equation}
For each basis state, the edges $e_1,\ldots,e_4$ of $p$ carry the group elements $g_1,\ldots,g_4$ and are oriented as in \eqref{eq:AvBp}. The gauge-fixed edge $e_3$ carries trivial irrep label $\mu_3=1$. 

All four edges of $\partial p$ can carry multiplication errors, and we can correct them in parallel with ancillas for the plaquettes on the outside, whose support doesn't include any other (multiplication) errors. 

The charge correction instead proceeds in two steps:
\begin{enumerate}
\item At the marked vertices $v_+(e_2)$ and $v_+(e_4)$, the only edges carrying irrep errors are $e_2$ and $e_4$, since $\mu_3=1$, so we can transfer their charges to ancillas and correct them.
\item Moving left, $e_1$ is now the only edge of $\star v_+(e_1)$ carrying an irrep error, and its charge is transferred at the top-left corner $v_+(e_1)$ of $p$.
\end{enumerate}

\paragraf{Two Plaquettes, Fig.~\ref{fig:gauge-fix}(d).}
Finally, the two plaquette configuration is 
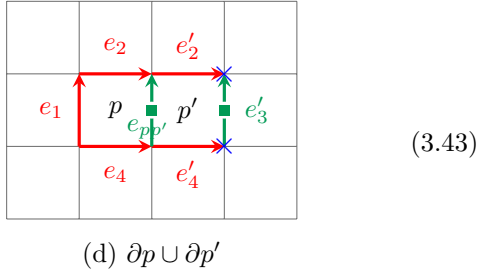
\begin{equation}
    \begin{split}
      \begin{tikzpicture}[scale=1.6,>=stealth]
        \draw[step=0.6,black!55] (0,0) grid (2.4,1.8);
        \draw[red,very thick,->] (0.6,1.2) -- (1.2,1.2);
        \draw[red,very thick,->] (1.2,1.2) -- (1.8,1.2);
        \draw[ForestGreen,very thick,->] (1.8,0.6) -- (1.8,1.2);
        \draw[red,very thick,->] (0.6,0.6) -- (1.2,0.6);
        \draw[red,very thick,->] (1.2,0.6) -- (1.8,0.6);
        \draw[red,very thick,->] (0.6,0.6) -- (0.6,1.2);
        \draw[ForestGreen,very thick,->] (1.2,0.6) -- (1.2,1.2);
        \fill[white] (1.2,0.9) circle (0.07);
        \fill[ForestGreen] (1.2,0.9) +(-0.045,-0.045) rectangle +(0.045,0.045);
        \fill[white] (1.8,0.9) circle (0.07);
        \fill[ForestGreen] (1.8,0.9) +(-0.045,-0.045) rectangle +(0.045,0.045);
        \node[blue,font=\large] at (1.8,1.2) {$\times$};
        \node[blue,font=\large] at (1.8,0.6) {$\times$};
        \node[red] at (0.37,0.9) {$e_1$};
        \node[red] at (0.9,1.43) {$e_2$};
        \node[red] at (0.9,0.37) {$e_4$};
        \node[red] at (1.5,1.43) {$e'_2$};
        \node[red] at (1.5,0.37) {$e'_4$};
         \node[ForestGreen,anchor=north] at (1.17,0.9) {$e_{pp'}$};
        \node[ForestGreen] at (2.07,0.9) {$e'_3$};
        \node at (0.9,0.9) {$p$};
        \node at (1.5,0.9) {$p'$};
        \node at (1.2,-0.3) {(d) $\partial p\cup\partial p'$};
      \end{tikzpicture}
    \end{split}
  \end{equation}
For the flux correction, we proceed as follows:
\begin{enumerate}
    \item As in the previous example, we we can correct all multiplication errors {other than that of $e_{pp'}$} in parallel with ancillas for the plaquettes other than $p,p'$, whose support doesn't include any other (multiplication) errors. 
    \item Now, $e'_2$, $e'_3$ and $e'_4$ carry no multiplication errors, so we correct the shared edge $e_{pp'}$ with an ancilla for $p'$ with the basepoint $v_-(e_{pp'})$, .
\end{enumerate}
The gauge-fixed edges $e_{pp'}$ and $e'_3$ carry trivial irrep labels. The charge correction can therefore proceed moving along the direct path between the marked vertices (to the right of $p'$ as follows:
\begin{enumerate}
\item At the marked vertices $v_+(e'_2)$ and $v_+(e'_4)$, transfer the charges of $e'_2$ and $e'_4$ to ancillas and remove them.
\item Moving left, $e_2$ and $e_4$ are now the only edges with irrep errors in the stars of $v_+(e_2)$ and $v_+(e_4)$. We can therefore transfer the charges of $e_2$ and $e_4$ to ancillas and remove them.
\item Finally, we transfer the charge of $e_1$ at $v_+(e_1)$, the top-left corner of $p$.
\end{enumerate}

\paragraf{Separated Star and Plaquette.}
Error configurations can contain both star and plaquettes. The simplest configuration of this type is when these are separated as e.g. here: 
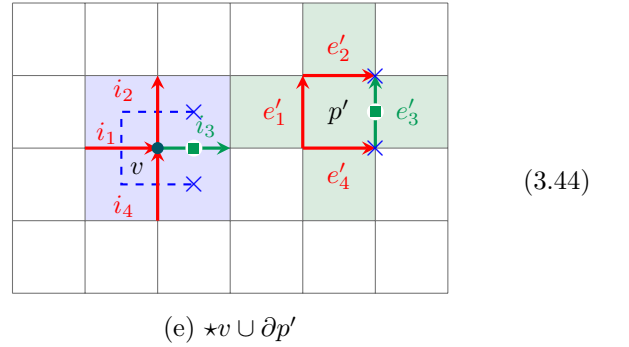
\begin{equation} \label{eq:separated-star-plaquette}
  \begin{split}
    \begin{tikzpicture}[scale=1.6,>=stealth]
      \fill[blue!12] (0.6,0) rectangle (1.8,1.2);
      \fill[ForestGreen!14] (1.8,0.6) rectangle (3.6,1.2);
      \fill[ForestGreen!14] (2.4,0) rectangle (3.0,1.8);
      \draw[step=0.6,black!55] (0,-0.6) grid (3.6,1.8);
      \draw[blue,dashed,thick] (1.5,0.9) -- (0.9,0.9) -- (0.9,0.3) -- (1.5,0.3);
      \draw[red,very thick,->] (0.6,0.6) -- (1.2,0.6);
      \draw[red,very thick,->] (1.2,0.6) -- (1.2,1.2);
      \draw[red,very thick,->] (1.2,0) -- (1.2,0.6);
      \draw[ForestGreen,very thick,->] (1.2,0.6) -- (1.8,0.6);
      \fill[blue(munsell)] (1.2,0.6) circle (0.05);
      \fill[white] (1.5,0.6) circle (0.07);
      \fill[ForestGreen] (1.5,0.6) +(-0.045,-0.045) rectangle +(0.045,0.045);
      \node[blue,font=\large] at (1.5,0.9) {$\times$};
      \node[blue,font=\large] at (1.5,0.3) {$\times$};
      \node[red] at (0.78,0.72) {$i_1$};
      \node[red] at (0.92,1.08) {$i_2$};
      \node[ForestGreen] at (1.6,0.78) {$i_3$};
      \node[red] at (0.92,0.12) {$i_4$};
      \node at (1.03,0.43) {$v$};
      \draw[red,very thick,->] (2.4,0.6) -- (3.0,0.6);
      \draw[red,very thick,->] (2.4,1.2) -- (3.0,1.2);
      \draw[red,very thick,->] (2.4,0.6) -- (2.4,1.2);
      \draw[ForestGreen,very thick,->] (3.0,0.6) -- (3.0,1.2);
      \fill[white] (3.0,0.9) circle (0.07);
      \fill[ForestGreen] (3.0,0.9) +(-0.045,-0.045) rectangle +(0.045,0.045);
      \node[blue,font=\large] at (3.0,1.2) {$\times$};
      \node[blue,font=\large] at (3.0,0.6) {$\times$};
      \node[red] at (2.17,0.9) {$e'_1$};
      \node[red] at (2.7,1.43) {$e'_2$};
      \node[ForestGreen] at (3.27,0.9) {$e'_3$};
      \node[red] at (2.7,0.37) {$e'_4$};
      \node at (2.7,0.9) {$p'$};
      \node at (1.8,-0.9) {{(e)} $\star v\cup\partial p'$};
    \end{tikzpicture}
  \end{split}
\end{equation}
The gauge-fixed edges are $i_3$, with $h_3=\id$, and $e'_3$, with $\mu'_3=1$, and the dashed line is the dual path between the marked plaquettes $p_\nearrow(v)$ and $p_\searrow(v)$, which crosses $i_2$, $i_1$ and $i_4$. 
The check region of $\star v$ consists of the four plaquettes with corner $v$, shaded in blue, and of the vertices of its edges, while the check region of $\partial p'$ consists of $p'$, its four neighboring plaquettes, shaded in green, and the four corners of $p'$. 

Although the check region \eqref{eq:check-region} of $i_3$ overlaps with that of $e_1'$, we call the configuration \eqref{eq:separated-star-plaquette} ``separated'' star and plaquette since the errored edges $i_3$ and $e'_1$ are separated by a clean edge and we will show that the errors on $\star v$ and $\partial p'$ can be corrected independently and in parallel. 

The flux correction stage is analogous to the the ones for the previous configurations:
\begin{enumerate}
\item {At the marked plaquettes $p_\nearrow(v)$ and $p_\searrow(v)$, transfer the multiplication errors of $i_2$ and $i_4$ to ancillas and remove them, as for the single star. In parallel, transfer the multiplication errors of $e'_1,\ldots,e'_4$ at the four plaquettes neighboring $p'$, each of which contains a single edge of $\partial p'$, like for the single plaquette.}
\item {Moving from right to left along the dual path, $i_1$ is now the only edge of $\partial p_\nwarrow(v)$ carrying a multiplication error, and it is corrected at $p_\nwarrow(v)$ with the basepoint $v_-(i_1)$.}
\end{enumerate}
{The charge correction stage consists of:}
\begin{enumerate}
\item {Transfer the charges of $i_1,\ldots,i_4$ to ancillas at their endpoints $v_-(i_1)$, $v_+(i_2)$, $v_+(i_3)$ and $v_-(i_4)$ other than $v$, whose stars contain no other error, as for the single star. In parallel, transfer the charges of $e'_2$ and $e'_4$, to ancillas at the marked vertices $v_+(e'_2)$ and $v_+(e'_4)$,  which are the only irrep errors on these stars since $\mu'_3=1$.}
\item {Moving left, $e'_1$ is now the only edge of $\star v_+(e'_1)$ carrying an irrep error, and its charge is transferred at the top-left corner $v_+(e'_1)$ of $p'$.}
\end{enumerate}
{Both stages thus take two transfer rounds, and the protocol uses $7+7=14$ ancillas, in agreement with \eqref{eq:gauge-fixed-count} for $N=8$ and $|\vec i_A|=|\vec i_B|=1$.}

\paragraf{Combined Star and Plaquette.}
Finally, consider when there are both star and plaquette errors that overlap. E.g. from the configuration \eqref{eq:separated-star-plaquette},
we now remove the clean edge between $i_3$ and $e'_1$, so that the bottom-left corner of $p'$ is the endpoint $v_+(i_3)$ of the gauge-fixed edge $i_3$:
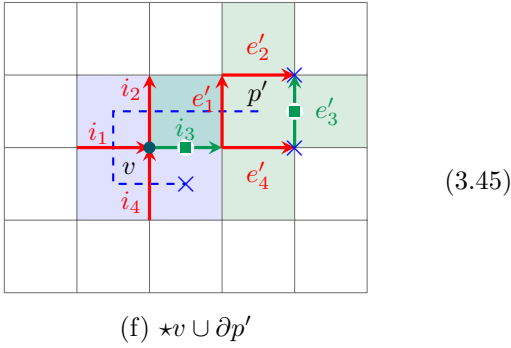
\begin{equation}
  \begin{split}
    \begin{tikzpicture}[scale=1.6,>=stealth]
      \fill[blue!12] (0.6,0) rectangle (1.8,1.2);
      \fill[ForestGreen!14] (1.2,0.6) rectangle (3.0,1.2);
      \fill[ForestGreen!14] (1.8,0) rectangle (2.4,1.8);
      \fill[teal!28] (1.2,0.6) rectangle (1.8,1.2);
      \draw[step=0.6,black!55] (0,-0.6) grid (3.0,1.8);
      \draw[blue,dashed,thick] (2.1,0.9) -- (0.9,0.9) -- (0.9,0.3) -- (1.5,0.3);
      \draw[red,very thick,->] (0.6,0.6) -- (1.2,0.6);
      \draw[red,very thick,->] (1.2,0.6) -- (1.2,1.2);
      \draw[red,very thick,->] (1.2,0) -- (1.2,0.6);
      \draw[ForestGreen,very thick,->] (1.2,0.6) -- (1.8,0.6);
      \fill[blue(munsell)] (1.2,0.6) circle (0.05);
      \fill[white] (1.5,0.6) circle (0.07);
      \fill[ForestGreen] (1.5,0.6) +(-0.045,-0.045) rectangle +(0.045,0.045);
      \node[blue,font=\large] at (1.5,0.3) {$\times$};
      \node[red] at (0.78,0.72) {$i_1$};
      \node[red] at (1.05,1.08) {$i_2$};
      \node[ForestGreen] at (1.5,0.75) {$i_3$};
      \node[red] at (1.05,0.15) {$i_4$};
      \node at (1.03,0.43) {$v$};
      \draw[red,very thick,->] (1.8,0.6) -- (2.4,0.6);
      \draw[red,very thick,->] (1.8,1.2) -- (2.4,1.2);
      \draw[red,very thick,->] (1.8,0.6) -- (1.8,1.2);
      \draw[ForestGreen,very thick,->] (2.4,0.6) -- (2.4,1.2);
      \fill[white] (2.4,0.9) circle (0.07);
      \fill[ForestGreen] (2.4,0.9) +(-0.045,-0.045) rectangle +(0.045,0.045);
      \node[blue,font=\large] at (2.4,1.2) {$\times$};
      \node[blue,font=\large] at (2.4,0.6) {$\times$};
      \node[red] at (1.65,1.0) {$e'_1$};
      \node[red] at (2.1,1.43) {$e'_2$};
      \node[ForestGreen] at (2.67,0.9) {$e'_3$};
      \node[red] at (2.1,0.37) {$e'_4$};
      \node at (2.1,1.05) {$p'$};
      \node at (1.5,-0.9) {{(f)} $\star v\cup\partial p'$};
    \end{tikzpicture}
  \end{split}
\end{equation}
{The edge supports $\star v$ and $\partial p'$ are still disjoint, so the gauge fixing is unchanged, with $h_3=\id$ and $\mu'_3=1$, and so is the count \eqref{eq:gauge-fixed-count}. Their check regions, however, now overlap: the plaquette $p_\nearrow(v)$, shaded in teal, contains both $i_2$ and $e'_1$, and the star of the shared vertex $v_+(i_3)=v_-(e'_1)=v_-(e'_4)$ contains $i_3$, $e'_1$ and $e'_4$. Theorem~\ref{thm:KLdegenerate} still applies, since the edges that can carry multiplication errors contain no complete star and the edges that can carry irrep errors contain no complete plaquette. However, $p_\nearrow(v)$ is no longer a single-error check so it no longer marked, and the two (vertex and plaquette) supports can no longer be corrected independently. The dashed dual path now runs from $p'$ to $p_\searrow(v)$ and crosses $e'_1$, $i_2$, $i_1$ and $i_4$. The flux correction stage is:}
\begin{enumerate}
\item {At the marked plaquette $p_\searrow(v)$, transfer the multiplication error of $i_4$ to an ancilla and remove it. In parallel, transfer the multiplication errors of $e'_2$, $e'_3$ and $e'_4$ at the plaquettes above, to the right of and below $p'$, each of which contains a single edge carrying a multiplication error.}
\item {Now $e'_1$ is the only edge of $\partial p'$ carrying a multiplication error, and $i_1$ is the only one of $\partial p_\swarrow(v)$. They are corrected in parallel, at $p'$ with the basepoint $v_-(e'_1)=v_+(i_3)$ and at $p_\swarrow(v)$ with the basepoint $v_-(i_1)$.}
\item {Finally, $i_2$ is the only remaining edge carrying a multiplication error, and it is corrected at $p_\nearrow(v)$ with the basepoint $v_-(i_2)=v$.}
\end{enumerate}
{Along the dual path, the multiplication errors are thus removed from both ends towards $i_2$, whose correction has to wait for that of $e'_1$. The charge correction stage is:}
\begin{enumerate}
\item {Transfer the charges of $i_1$, $i_2$ and $i_4$ at $v_-(i_1)$, $v_+(i_2)$ and $v_-(i_4)$, whose stars contain no other error. In parallel, at the marked vertices $v_+(e'_2)$ and $v_+(e'_4)$, transfer the charges of $e'_2$ and $e'_4$, since $\mu'_3=1$.}
\item {Moving left, $e'_1$ is now the only edge of $\star v_+(e'_1)$ carrying an irrep error, and its charge is transferred at the top-left corner $v_+(e'_1)$ of $p'$.}
\item {Finally, the star of the shared vertex $v_+(i_3)$ now contains only the irrep error on $i_3$, since $e'_1$ and $e'_4$ have been corrected, and the charge of $i_3$ is transferred there.}
\end{enumerate}
{Compared with the separated configuration, the overlap of the check regions costs one additional step in each stage, since $i_2$ has to wait for $e'_1$ and $i_3$ has to wait for $e'_1$ and $e'_4$, while the number of ancillas, $14$, is unchanged.}

\paragraf{Result.} In each case every free label of \eqref{eq:gauge-fixed-family} is moved onto its own ancilla, so the protocol uses $2N-|\vec i_A|-|\vec i_B|$ ancillas, in agreement with the gauge-fixed error space dimension \eqref{eq:gauge-fixed-count}. 
By \eqref{flux-transfer-full} and \eqref{CL-U-result} the data qudits are returned to $\ket\psi$ and the measurements of the ancillas give the error labels $(\vec h,\vec\mu)$.

\subsection{Ribbon Errors and Correction}
\label{sec:RibEr}

An important type of error arises from ribbon operators. The details are tedious and not containing any new insights compared to what we have discussed already. Nevertheless these are important types of errors so we discuss them in App.~\ref{app:ribbon}.
A ribbon error carries a single pair of labels $(h,\mu)$ on an extended 1d support on the lattice or dual lattice, and is detected only by the check violations at its two end-points. 
It is essentially an extended single-edge error. A ribbon error basis is determined in App.~\ref{app:REB} and we show the orthogonality on the code space.In App.~\ref{app:REC} the ancilla-based recovery is given, following the single-edge constructions above.

\section{Conclusion}

Non-abelian surface codes have recently seen resurged interest due to their favorably universal quantum computing properties. Unlike the surface codes for abelian groups, the QEC for non-abelian groups is still poorly understood. In this paper we address two questions that bring us closer to a universally applicable QEC protocol for such codes: we determine an orthogonal error basis, and we devise an exact recovery protocol that moves error configurations in charge neutral clusters, onto ancillas and projects them out. 

One might worry that because the error transfer must initially be performed only on the boundary of a cluster, the correction scales linearly with the size of the cluster. Nevertheless, one might hope that because large clusters are very rare, it might still be possible to obtain a threshold if we only shrink the cluster by a constant amount in each error correction round rather than performing the full sequential unitary to completely shrink the cluster. We leave this numerical verification as well as finding protocols that also deal with measurement errors, to future work. 

An alternative approach, which we will explore in future work, is to evade even the initial identification of the charge neutral error clusters, and to apply the error transfer protocol to the entire code, iteratively. Such a parallel QEC protocol may be even more favorable and gets around the assumption that neutral error clusters have been identified.

\subsection*{Acknowledgments}
\noindent
We thank Noa Feldman, Julio Magdalena de la Fuente, Rajath Radhakrishnan, Steve Simon and Dominic Williamson for discussions of related matters. Anthropic's Claude Fable 5.1 and Opus 5.5 and ChatGPT Astra were used in parts of this work to improve formulations, do numerical and theorem checks. 
The work of SSN and AW is supported by the UKRI Frontier Research Grant, underwriting the ERC Advanced Grant ``Generalized Symmetries in Quantum Field Theory and Quantum Gravity''.

\appendix

\section{Proofs}

\subsection{Quantum Error Correction Conditions}

Here we review the Knill-Laflamme quantum error correction (QEC) conditions \cite{Knill:1996ny,Bennett:1996gf}, see also \cite[Theorem~10.1]{Nielsen:2012yss}.

\begin{theorem}[QEC conditions] \label{thm:Knill-Laflamme}
Let $C$ be a quantum code with code-space projector $P$, and let $\mathcal E$ be a quantum operation with Kraus operators $\{\wt{E}_\mu\}$. A necessary and sufficient condition for the existence of a recovery operation $\mathcal R$ correcting $\mathcal E$ on $C$ is
\be\label{eq:Knill-Laflamme}
P\wt{E}_\mu^\dagger\wt{E}_\nu P=\alpha_{\mu\nu}P\,.
\ee
Here $\alpha=(\alpha_{\mu\nu})$ is positive semidefinite and Hermitian.
\end{theorem}

\begin{proof}[Proof]
Since $\alpha$ is Hermitian, it is diagonalizable by a unitary $u$, such that
\be
u^\dagger\alpha u=d\,,
\ee
with $d$ diagonal. The conditions \eqref{eq:Knill-Laflamme} for the rotated Kraus operators
\be
E_\kappa=\sum_\mu u_{\mu\kappa}\wt{E}_\mu\,,
\ee
then become
\be\label{eq:E_orth}
P E_\kappa^\dagger E_\lambda P=d_{\kappa\kappa}\delta_{\kappa\lambda}P\,.
\ee
Restricting $\kappa$ to the indices such that $d_{\kappa\kappa}>0$, and for each such $\kappa$ the error syndrome projector are defined as
\be \label{eq:KL-diag}
P_\kappa=\frac{E_\kappa P E_\kappa^\dagger}{d_{\kappa\kappa}}\,.
\ee
By \eqref{eq:E_orth} the $P_\kappa$ are mutually orthogonal projectors onto the error subspaces and $E_\kappa/\sqrt{d_{\kappa\kappa}}$ is an isometry on the code space, which extends to a unitary $U_\kappa$ on the full Hilbert space\footnote{$U_\kappa$ exists by the matrix polar decomposition of $E_\kappa P$ and \eqref{eq:E_orth}.} such that:
\be\label{eq:EP_UP}
E_\kappa P=\sqrt{d_{\kappa\kappa}}\,U_\kappa P\,.
\ee
The recovery channel measures the syndromes $\{P_\kappa\}$ and applies the corresponding $U_\kappa^\dagger$,
\be\ba\label{eq:recovery-channel}
\mathcal R(\sigma)&=\sum_{\kappa} U_\kappa^\dagger P_\kappa\sigma P_\kappa U_\kappa+P_\perp\sigma P_\perp\,,\cr
P_\perp&:=\bbI-\sum_\kappa P_\kappa\,.
\ea\ee
The complementary projector $P_\perp$ vanishes on the corrupted code space and completes $\mathcal R$ to a trace-preserving channel.
{In the following when we refer to projective measurement, we imply that such a complementary term can always be added.}
Then for every code state $\rho=P\rho P$,\footnote{Using \eqref{eq:E_orth}-\eqref{eq:EP_UP}, one has:
\be
U_\kappa^\dagger P_\kappa E_\lambda P=\sqrt{d_{\kappa\kappa}}\,\delta_{\kappa\lambda}P\,.
\ee
}
the recovery channel is
\be\ba
\mathcal R(\mathcal E(\rho))&=\sum_{\kappa,\lambda} U_\kappa^\dagger P_\kappa E_\lambda\rho E_\lambda^\dagger P_\kappa U_\kappa=\left(\sum_\kappa d_{\kappa\kappa}\right)\rho\,.
\ea\ee
The factor $\sum_\kappa d_{\kappa\kappa}=\Tr[\mathcal E(\rho)]$ is independent of $\rho$ and any trace-preserving error channel $\mathcal E$ is corrected exactly. Necessity is proven in e.g. \cite[Theorem~10.1]{Nielsen:2012yss}.
\end{proof}

\paragraf{Code Distance.} The code distance is the smallest support carrying an operator that is not correctable on the code space by Thm.~\ref{thm:Knill-Laflamme}:
\be\label{eq:distance}
d=\min\left\{|\vec i|:\exists\, E_{\vec i},\ PE_{\vec i}P\notin\bC P\right\}\,.
\ee
E.g. on a square patch of linear size $L$ with opposite smooth and rough boundaries, as in \cite{Huang:2025ump,Manjunath:2026swm}, $d=L$, while for other boundary conditions $d$ depends on the boundary subgroups but is still $O(L)$ due to the topological property of the code \cite{Cui:2019lvb}. Conversely, all operators on a fixed support with {$|\vec i|<d$} obey \eqref{eq:Knill-Laflamme} and are therefore correctable.

\subsection{Proofs for Section \ref{sec:errors}} \label{app:proof2}

\begin{lemma}[Single-site algebra]\label{lemma:single-site}
The operators \eqref{eq:LR_def}--\eqref{eq:ZGamma_def} satisfy \cite{Chung:2025ulc}
\be\label{eq:ZL_comm}
Z^\Gamma_{ab}L^g=\sum_{c=1}^{d_\Gamma}\Gamma(g)_{ac}L^gZ^\Gamma_{cb}\,,
\ee
and
\be\label{eq:ZR_comm}
Z^\Gamma_{ab}R^g=\sum_{c=1}^{d_\Gamma}R^gZ^\Gamma_{ac}\Gamma(g^{-1})_{cb}\,.
\ee
Equivalently,
\be
L^gZ^\Gamma_{ab}=\sum_{c=1}^{d_\Gamma}\Gamma(g^{-1})_{ac}Z^\Gamma_{cb} L^g
\label{eq:ZL_comm2}
\ee
and
\be
R^gZ^\Gamma_{ab}=\sum_{c=1}^{d_\Gamma}\Gamma(g)_{cb}Z^\Gamma_{ac} R^g
\label{eq:ZR_comm2}
\ee
\end{lemma}
\begin{proof}
For any $h\in G$, $\Gamma(gh)=\Gamma(g)\Gamma(h)$ gives
\be\ba
Z^\Gamma_{ab}L^g\ket{h}&=\Gamma(gh)_{ab}\ket{gh}=\sum_{c=1}^{d_\Gamma}\Gamma(g)_{ac}\Gamma(h)_{cb}\ket{gh}\cr
&=\sum_{c=1}^{d_\Gamma}\Gamma(g)_{ac}L^gZ^\Gamma_{cb}\ket{h}\,,
\ea\ee
and likewise
\be\ba
Z^\Gamma_{ab}R^g\ket{h}&=\Gamma(hg^{-1})_{ab}\ket{hg^{-1}}\cr
&=\sum_{c=1}^{d_\Gamma}\Gamma(h)_{ac}\Gamma(g^{-1})_{cb}\ket{hg^{-1}}\cr
&=\sum_{c=1}^{d_\Gamma}R^gZ^\Gamma_{ac}\Gamma(g^{-1})_{cb}\ket{h}\,.
\ea\ee
\end{proof}

\begin{proof}[Proof of Proposition~\ref{eq:prop_Avmu}]
By \eqref{eq:Avmu_def} and $A_v^gA_v^h=A_v^{gh}$,
\be
A_v^{(\Gamma,a,b)}A_v^{(\Lambda,c,d)}=\frac{d_\Gamma d_\Lambda}{|G|^2}\sum_{g,h\in G}\Gamma(g)_{ab}\Lambda(h)_{cd}A_v^{gh}\,.
\ee
Setting $k=gh$ gives
\be
A_v^{(\Gamma,a,b)}A_v^{(\Lambda,c,d)}=\frac{d_\Gamma d_\Lambda}{|G|^2}\sum_{k,g\in G}\Gamma(g)_{ab}\Lambda(g^{-1}k)_{cd}A_v^k\,.
\ee
Inserting the expansion
\be
\Lambda(g^{-1}k)_{cd}=\sum_{t=1}^{d_\Lambda}\Lambda(g^{-1})_{ct}\Lambda(k)_{td}\,,
\ee
into Schur orthogonality,
\be
\sum_{g\in G}\Gamma(g)_{ab}\Lambda(g^{-1})_{ct}=\frac{|G|}{d_\Gamma}{\delta_{\Gamma,\Lambda}\delta_{b,c}\delta_{a,t}}\,,
\ee
gives
\be
\sum_{g\in G}\Gamma(g)_{ab}\Lambda(g^{-1}k)_{cd}=\frac{|G|}{d_\Gamma}{\delta_{\Gamma,\Lambda}\delta_{b,c}}\Gamma(k)_{ad}\,,
\ee
and hence \eqref{eq:Avmu_comm}. {Finally, $(A_v^g)^\dagger=A_v^{g^{-1}}$ unitarity, $\Gamma(g)^*_{ab}=\Gamma(g^{-1})_{ba}$ and the substitution $g\to g^{-1}$ give \eqref{eq:Avmu-adjoint}.}
\end{proof}

\subsection{Proofs for Section \ref{sec:errorcorrection}}
\label{app:proof3}

In this appendix, we prove various useful identifies that were stated in the main text.
\begin{lemma}\label{lemma:CL}
\be
CL_{i,\mathrm{va}}\Bigl(\sum_{c=1}^{d_\Gamma}Z^\Gamma_{ac,i}Z^\Gamma_{cb,\mathrm{va}}\Bigr)CL_{i,\mathrm{va}}^\dagger=Z^\Gamma_{ab,\mathrm{va}}\,.
\ee
and
\be
CL^\dagger_{i,\mathrm{va}}\Bigl(\sum_{c=1}^{d_\Gamma}Z^\Gamma_{cb,i}Z^{\Gamma^*}_{ca,\mathrm{va}}\Bigr)CL_{i,\mathrm{va}}=Z^{\Gamma^*}_{ba,\mathrm{va}}\,.
\ee
\end{lemma}
\begin{proof}
    The operator $\sum_{c=1}^{d_\Gamma}Z^\Gamma_{ac,i}Z^\Gamma_{cb,\mathrm{va}}$ is diagonal in the group basis, with eigenvalue
\be\ba 
\label{O-eigenvalue}
&\Bigl(\sum_{c=1}^{d_\Gamma}Z^\Gamma_{ac,i}Z^\Gamma_{cb,\mathrm{va}}\Bigr)\ket x_i\ket y_{\mathrm{va}}\cr 
&=\sum_{c=1}^{d_\Gamma}\Gamma(x)_{ac}\Gamma(y)_{cb}\ket x_i\ket y_{\mathrm{va}}\cr 
&=\Gamma(xy)_{ab}\ket x_i\ket y_{\mathrm{va}}\,.
\ea\ee
Conjugation by $CL_{i,\mathrm{va}}$ therefore removes the dependence on $x$,
\be
\label{CL-conjugation}
\ba 
&CL_{i,\mathrm{va}}\Bigl(\sum_{c=1}^{d_\Gamma}Z^\Gamma_{ac,i}Z^\Gamma_{cb,\mathrm{va}}\Bigr)CL_{i,\mathrm{va}}^\dagger\ket x_i\ket y_{\mathrm{va}}\cr 
&=CL_{i,\mathrm{va}}\Bigl(\sum_{c=1}^{d_\Gamma}Z^\Gamma_{ac,i}Z^\Gamma_{cb,\mathrm{va}}\Bigr)\ket x_i\ket{x^{-1}y}_{\mathrm{va}}\cr 
&=\Gamma(y)_{ab}\ket x_i\ket y_{\mathrm{va}}\,.
\ea
\ee

For the second identity, the operator $\sum_{c=1}^{d_\Gamma}Z^\Gamma_{cb,i}Z^{\Gamma^*}_{ca,\mathrm{va}}$ is diagonal in the group basis, which acts on the state $\ket x_i\ket y_{\mathrm{va}}$ with  eigenvalue 
\be 
\sum_{c=1}^{d_\Gamma}\Gamma(y)^*_{ca}\Gamma(x)_{cb}=\sum_{c=1}^{d_\Gamma}\Gamma(y^{-1})_{ac}\Gamma(x)_{cb}=\Gamma(y^{-1}x)_{ab} \,,
\ee 
where we used that $\Gamma(y)$ is unitary. 
Since $CL_{i,\mathrm{va}}\ket x_i\ket y_{\mathrm{va}}=\ket x_i\ket{xy}_{\mathrm{va}}$, conjugation by $CL_{i,\mathrm{va}}^\dagger$ gives
\be
\label{CL-conjugation-dagger}
\ba
&CL^\dagger_{i,\mathrm{va}}\Bigl(\sum_{c=1}^{d_\Gamma}Z^\Gamma_{cb,i}Z^{\Gamma^*}_{ca,\mathrm{va}}\Bigr)CL_{i,\mathrm{va}}\ket x_i\ket y_{\mathrm{va}}\cr
&=\Gamma\bigl((xy)^{-1}x\bigr)_{ab}\ket x_i\ket y_{\mathrm{va}}=\Gamma(y^{-1})_{ab}\ket x_i\ket y_{\mathrm{va}}\,.
\ea
\ee
Thus, $\Gamma(y^{-1})_{ab}=\Gamma(y)^*_{ba}=\Gamma^*(y)_{ba}$ is the eigenvalue of $Z^{\Gamma^*}_{ba,\mathrm{va}}$, which proves the second identity.

\end{proof}

\section{Ribbon Errors}
\label{app:ribbon}

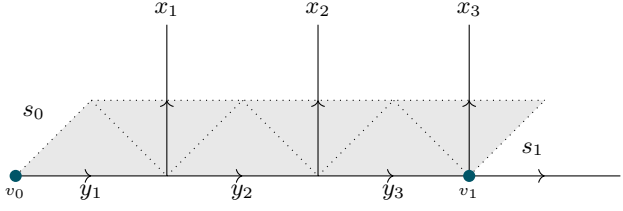
\begin{figure}
\centering
\begin{tikzpicture}[scale=2]
\fill[gray!18] (0,0) -- (1,0) -- (0.5,0.5) -- cycle;
\fill[gray!18] (0.5,0.5) -- (1,0) -- (1.5,0.5) -- cycle;
\fill[gray!18] (1,0) -- (2,0) -- (1.5,0.5) -- cycle;
\fill[gray!18] (1.5,0.5) -- (2,0) -- (2.5,0.5) -- cycle;
\fill[gray!18] (2,0) -- (3,0) -- (2.5,0.5) -- cycle;
\fill[gray!18] (2.5,0.5) -- (3,0) -- (3.5,0.5) -- cycle;
\draw[black,->-] (0,0) -- (1,0);
\draw[black,->-] (1,0) -- (2,0);
\draw[black,->-] (2,0) -- (3,0);
\draw[black,->-] (3,0) -- (4,0);
\draw[black,->-] (1,0) -- (1,1);
\draw[black,->-] (2,0) -- (2,1);
\draw[black,->-] (3,0) -- (3,1);
\draw[black,dotted] (0,0) -- (0.5,0.5);
\draw[black,dotted] (0.5,0.5) -- (1,0);
\draw[black,dotted] (1,0) -- (1.5,0.5);
\draw[black,dotted] (1.5,0.5) -- (2,0);
\draw[black,dotted] (2,0) -- (2.5,0.5);
\draw[black,dotted] (2.5,0.5) -- (3,0);
\draw[black,dotted] (3,0) -- (3.5,0.5);
\draw[black,dotted] (0.5,0.5) -- (3.5,0.5);
\fill[blue(munsell)] (0,0) circle (0.04);
\fill[blue(munsell)] (3,0) circle (0.04);
\node at (0.12,0.42) {${s_0}$};
\node at (3.42,0.18) {${s_1}$};
\node[below] at (0.5,0) {${y_1}$};
\node[below] at (1.5,0) {${y_2}$};
\node[below] at (2.5,0) {${y_3}$};
\node[above] at (1,1) {${x_1}$};
\node[above] at (2,1) {${x_2}$};
\node[above] at (3,1) {${x_3}$};
\node[below] at (0,-0.02) {$_{v_0}$};
\node[below] at (3,-0.02) {$_{v_1}$};
\end{tikzpicture}
\caption{Rbbon operators: the figure shows an open ribbon $\xi$ with $n=3$ stretching from $s_0=(v_0,p_0)$ to $s_1=(v_1,p_1)$. 
Its dual path crosses the upward edges $x_1,\dots,x_n$ and its direct path $\gamma$ follows the rightward edges $y_1,\dots,y_n$. These $2n$ edges form $\operatorname{supp}\xi$. 
The end vertices  are $v_0$ and $v_1$ and every check away from them commutes with the ribbon algebra.}
\label{fig:ribbon}
\end{figure}

Ribbon errors are essentially extended single-edge errors: a ribbon error carries one pair of labels $(h,\mu)$ as in \eqref{eq:Emui} on the support of a ribbon $\xi$, and it violates the checks only at the two end sites of $\xi$. 
The constructions of Secs.~\ref{sec:single-detect}, \ref{sec:KL}, \ref{sec:charge-transfer} and \ref{sec:flux-transfer} therefore carry over step by step, with the edge $i$ replaced by $\xi$, the multiplication $L_i^h$ by the ribbon operator $F_\xi^h$ and the irrep operator $Z^\Gamma_{ab,i}$ by a Wilson line, and we only indicate the modifications. We use the ribbon conventions and properties derived in \cite{Kitaev:1997wr,Bombin:2007qv,Beigi:2010htr,Song:2025gly,Huang:2025ump}, specialized to the square lattice with every edge oriented right or up.

\subsection{Ribbon Error Basis}
\label{app:REB}
Let us first summarize our conventions, see also Fig.~\ref{fig:ribbon}:
$\xi$ is a ribbon from the site $s_0=(v_0,p_0)$ to the site $s_1=(v_1,p_1)$. 
It is along a path $\gamma$ consisting of $n$ consecutive edges $y_1,\dots,y_n$ from $v_0$ to $v_1$. 
Its dual path crosses the $n$ upward edges $x_1,\dots,x_n$. 

The plaquettes at the ends of the ribbons 
are $p_0$ (above $y_1$ whose bottom-left corner is $v_0$) and $p_1$ (above the edge to the right of $v_1$, whose bottom-left corner is $v_1$) such that $x_1$ is the right edge of $p_0$ and $x_n$ is the left edge of $p_1$. 
The $2n$ edges $x_m,y_m$ form te support of $\xi$, and every check away from $s_0$ and $s_1$ commutes with the ribbon operators reviewed below.

\paragraf{Ribbon Operators.} We use $x_m$ and $y_m$ both for an edge and for the group label of its basis state. We set $\hat y_0:=\id$ and $\hat y_m:=y_1\cdots y_m$. The ribbon operators project the ordered product along $\gamma$ onto $\hat y_n=g$ and multiply each crossed edge by a transported conjugate of $h$,
\be\ba\label{eq:ribbon-op}
F_\xi^{h,g}&:=\sum_{y_1,\dots,y_n}\delta_{g,\hat y_n}\,\ketbra{y_1,\dots,y_n}\cr
&\qquad\ot\prod_{m=1}^nL_{x_m}^{\hat y_m^{-1}h\hat y_m}\,,\qquad F_\xi^h:=\sum_{g\in G}F_\xi^{h,g}\,.
\ea\ee
They obey $F_\xi^{h,g}F_\xi^{k,l}=\delta_{g,l}F_\xi^{hk,g}$ and $(F_\xi^{h,g})^\dagger=F_\xi^{h^{-1},g}$, so the $|G|^2$ operators $F_\xi^{h,g}$ span an algebra $\cF_\xi$ closed under products and adjoints, and $F_\xi^h$ is unitary with $F_\xi^hF_\xi^k=F_\xi^{hk}$. 

Every $A_v$ and $B_p$ not at $s_0$ or $s_1$ commutes with $\cF_\xi$, since $F_\xi^{h,g}$ preserves the flux of every plaquette between consecutive crossed edges and the vertex operators at the interior vertices of $\gamma$ leave $\hat y_n$ and the transported conjugates invariant. At the end sites, with the conventions of \eqref{eq:LR_def} and \eqref{eq:AvBp} and the clockwise holonomy of $p_j$ based at $v_j$  they satisfy \cite{Beigi:2010htr,Song:2025gly}
\be\ba\label{eq:ribbon-ends}
A_{v_0}^kF_\xi^{h,g}&=F_\xi^{khk^{-1},kg}A_{v_0}^k\,,\cr
A_{v_1}^kF_\xi^{h,g}&=F_\xi^{h,gk^{-1}}A_{v_1}^k\,,\cr
B_{p_0}^fF_\xi^{h,g}&=F_\xi^{h,g}B_{p_0}^{fh}\,,\cr
B_{p_1}^fF_\xi^{h,g}&=F_\xi^{h,g}B_{p_1}^{g^{-1}h^{-1}gf}\,.
\ea\ee

\paragraf{Ribbon Error Basis.} The Wilson line along $\gamma$ replaces the single-edge irrep operator \eqref{eq:ZGamma_def},
\be\ba\label{eq:ribbon-wilson}
Z^\Gamma_{ab}(\gamma)&:=\Gamma(\hat y_n)_{ab}\cr
&=\sum_{c_1,\dots,c_{n-1}}Z^\Gamma_{ac_1,y_1}Z^\Gamma_{c_1c_2,y_2}\cdots Z^\Gamma_{c_{n-1}b,y_n}\,,
\ea\ee
and, for $\mu=(\Gamma,a,b)$, the ribbon errors are
\be\ba\label{eq:ribbon-error-basis}
E_\xi^\mu&:=\sqrt{d_\Gamma}\,Z^\Gamma_{ab}(\gamma)\,,\cr
E_\xi^{h,\mu}&:=F_\xi^hE_\xi^\mu=\sqrt{d_\Gamma}\sum_{g\in G}\Gamma(g)_{ab}F_\xi^{h,g}\,,
\ea\ee
with $E_\xi^{\id,1}=\bbI$. This is \eqref{eq:Emui} with $L_i^h$ replaced by $F_\xi^h$ and $Z^\Gamma_{ab,i}$ by $Z^\Gamma_{ab}(\gamma)$, with a single pair of labels $(h,\mu)$ for the whole support.

\begin{proposition}[Ribbon error basis]\label{prop:ribbon-basis}
For fixed $\xi$ the $|G|^2$ operators $E_\xi^{h,\mu}$ form an orthogonal basis of $\cF_\xi$,
\be\label{eq:ribbon-orth}
\Tr\!\left[(E_\xi^{h,\mu})^\dagger E_\xi^{k,\nu}\right]=|G|^{2n}\delta_{h,k}\delta_{\mu,\nu}\,.
\ee
\end{proposition}
\begin{proof}
Since 
\be
\Tr L^a=|G|\delta_{a,\id}\,,\quad  |\{y:\hat y_n=g\}|=|G|^{n-1}\,,
\ee
equation \eqref{eq:ribbon-op} gives $\Tr F_\xi^{u,g}=|G|^{2n-1}\delta_{u,\id}$, so 
\be
\Tr[(F_\xi^{h,g})^\dagger F_\xi^{k,l}]=|G|^{2n-1}\delta_{h,k}\delta_{g,l}
\ee
by the algebra relations. Inserting \eqref{eq:ribbon-error-basis} and using Schur orthogonality as in \eqref{eq:setup-10} gives \eqref{eq:ribbon-orth}, and $|G|\sum_\Gamma d_\Gamma^2=|G|^2=\dim\cF_\xi$ gives completeness.
\end{proof}
\noindent The factor $|G|^{2n}$ replaces the $|G|$ of \eqref{eq:Emunu_delta} because $\cF_\xi$ is a $|G|^2$-dimensional subalgebra of the operators on the $2n$ edges of $\operatorname{supp}\xi$.

\paragraf{Flux and Charge Checks at the Ends.} As in \eqref{eq:local-checks} we collect the checks at the two end sites,
\be\ba\label{eq:ribbon-local-checks}
v(\xi)&:=\{v_0,v_1\}\,,\qquad p(\xi):=\{p_0,p_1\}\,,\cr
P_\xi&:=A_{v_0}A_{v_1}B_{p_0}B_{p_1}\,.
\ea\ee
Substituting \eqref{eq:ribbon-ends} into \eqref{eq:ribbon-error-basis} gives the ribbon analogue of \eqref{A-E-commutation},
\be\ba\label{eq:ribbon-endcov}
A_{v_1}^kE_\xi^{h,(\Gamma,a,b)}&=\sum_{c=1}^{d_\Gamma}\Gamma(k)_{cb}\,E_\xi^{h,(\Gamma,a,c)}A_{v_1}^k\,,\cr
A_{v_0}^kE_\xi^{h,(\Gamma,a,b)}&=\sum_{c=1}^{d_\Gamma}\Gamma(k^{-1})_{ac}\,E_\xi^{khk^{-1},(\Gamma,c,b)}A_{v_0}^k\,,
\ea\ee
so $\gamma$ behaves at $v_1$ as an edge entering $v_1$ and at $v_0$ as an edge leaving $v_0$, with $b$ the charge index at $v_1$ and $a$ that at $v_0$.
\begin{lemma}[Ribbon detectability]\label{lemma:ribbon-detectable}
Let $v_0\neq v_1$ and $p_0\neq p_1$, with all four checks untruncated. Then for $h\neq\id$
\be\label{eq:ribbon-PBEPB}
B_{p_0}F_\xi^{h,g}B_{p_0}=0\,,\qquad B_{p_0}E_\xi^{h,\mu}B_{p_0}=0\,,
\ee
for $\mu=(\Gamma,a,b)$ with $\Gamma\neq1$
\be\label{eq:ribbon-AvEAv}
A_{v_1}E_\xi^{\mu}A_{v_1}=0\,,
\ee
and hence $P_\xi E_\xi^{h,\mu}P_\xi=0$ for every $(h,\mu)\neq(\id,1)$.
\end{lemma}
\begin{proof}
This is Lemmas~\ref{lemma:PBEPB} and \ref{lemma:vertex-violations} for a ribbon. The third relation of \eqref{eq:ribbon-ends} with $f=\id$ gives $B_{p_0}F_\xi^{h,g}B_{p_0}=F_\xi^{h,g}B_{p_0}^{h}B_{p_0}^{\id}=0$ for $h\neq\id$, and $E_\xi^\mu$ is diagonal in the group basis, which gives \eqref{eq:ribbon-PBEPB}. Averaging the first line of \eqref{eq:ribbon-endcov} over $k$, Schur orthogonality $\sum_{k}\Gamma(k)_{cb}=|G|\delta_{\Gamma,1}$ gives \eqref{eq:ribbon-AvEAv}. 
Finally, each of the four projectors in $P_\xi$ commutes with the other three, so $P_\xi E_\xi^{h,\mu}P_\xi$ contains the factor $B_{p_0}E_\xi^{h,\mu}B_{p_0}$ for $h\neq\id$ and the factor $A_{v_1}E_\xi^{\mu}A_{v_1}$ for $h=\id$. Note that the diagonal $E_\xi^\mu$ does not commute with $A_{v_0}$, see the second line of \eqref{eq:ribbon-endcov}, but this is not needed.

\end{proof}

We can now prove a ribbon version of the Knill-Laflamme theorem:
\begin{theorem}[Ribbon Knill-Laflamme]\label{thm:KLribbon}
Let $\xi$ satisfy the hypotheses of Lemma~\ref{lemma:ribbon-detectable} and not host a logical operator. Then
\be\label{eq:KL-ribbon}
P_\xi(E_\xi^{h,\mu})^\dagger E_\xi^{k,\nu}P_\xi=\delta_{h,k}\delta_{\mu,\nu}P_\xi\,,
\ee
and the same holds with $P_\xi$ replaced by the global projector $P$, so the Knill-Laflamme matrix of the $|G|^2$ ribbon errors is the identity.
\end{theorem}
\begin{proof}
Since $\cF_\xi$ is closed under products and adjoints, $(E_\xi^{h,\mu})^\dagger E_\xi^{k,\nu}\in\cF_\xi$ can be expanded in the basis of Proposition~\ref{prop:ribbon-basis}, with the coefficient $\delta_{h,k}\delta_{\mu,\nu}$ of $E_\xi^{\id,1}=\bbI$ fixed by \eqref{eq:ribbon-orth}. Conjugation by $P_\xi$ annihilates every other term by Lemma~\ref{lemma:ribbon-detectable}. No bound on the length $n$ is needed, since every term is again a ribbon error on $\xi$. Since $P_\xi P=P$, the global identity follows.
\end{proof}

\subsection{Ribbon Correction Circuits}
\label{app:REC}

Ribbons with the same end sites that are related by a deformation through a region without excitations and without boundaries act identically on the code space \cite{Kitaev:1997wr,Bombin:2007qv,Beigi:2010htr} and belong to the same error sector. Ribbons in different relative homotopy classes can differ by a logical operator and are excluded by the hypothesis of the theorem.

\paragraf{Flux Transfer.} We correct a ribbon error $E_\xi^{h,\mu}\ket\psi$, with $\ket\psi$ a code state, by the two-stage protocol of Sec.~\ref{sec:gauge-fixed-correction}, first removing $h$ and then $\mu$. As in Sec.~\ref{sec:flux-transfer} we base the holonomy of $p_0$ at $v_0$ and traverse $\partial p_0$ along the orientation of the first edges of the ribbon, first along $y_1$ and then along $x_1$. With $g_1$ and $g_2$ the labels of the left and top edges of $p_0$ as in \eqref{eq:AvBp}, this holonomy is $y_1x_1g_2^{-1}g_1^{-1}$, and since $F_\xi^h$ maps $x_1$ to $y_1^{-1}hy_1x_1$ it multiplies the holonomy by $h$ from the left, exactly as $L_i^h$ in \eqref{B-L-commutation}. Denoting by $B_{p_0}^g$ the projector onto this holonomy, the controlled holonomy $V_{p_0}$ of \eqref{V-def} and \eqref{V-on-error} gives
\be\label{eq:ribbon-V-on-error}
V_{p_0}\ket\id_{\mathrm{pa}}\ot E_\xi^{h,\mu}\ket\psi=\ket h_{\mathrm{pa}}\ot E_\xi^{h,\mu}\ket\psi\,,
\ee
where we used that $E_\xi^\mu$ is diagonal in the group basis. The ribbon analogue of $CL_{\mathrm{pa},i}^\dagger$ in \eqref{CL-dagger} is the controlled inverse ribbon operator
\be\ba\label{eq:ribbon-CF-dagger}
CF_{\mathrm{pa},\xi}^\dagger&:=\sum_{f\in G}\ketbra f_{\mathrm{pa}}\ot F_\xi^{f^{-1}}\,,\cr
CF_{\mathrm{pa},\xi}^\dagger\,V_{p_0}\ket\id_{\mathrm{pa}}\ot E_\xi^{h,\mu}\ket\psi&=\ket h_{\mathrm{pa}}\ot E_\xi^{\mu}\ket\psi\,,
\ea\ee
which is unitary since $F_\xi^h$ is a unitary representation of $G$, and a measurement of the ancilla in the group basis identifies $h$. In terms of controlled multiplications \eqref{controlled-mult}, first write the partial products $\hat y_m$ onto the edges $y_m$ with $Y_\gamma:=CL_{y_{n-1},y_n}\cdots CL_{y_2,y_3}CL_{y_1,y_2}$, with the rightmost gate applied first. Then $L_{x_m}^{\hat y_m^{-1}f^{-1}\hat y_m}=CL_{y_m,x_m}^\dagger L_{x_m}^{f^{-1}}CL_{y_m,x_m}$, so
\be\label{eq:ribbon-CF-circuit}
CF_{\mathrm{pa},\xi}^\dagger=Y_\gamma^\dagger\Bigl[\prod_{m=1}^nCL_{y_m,x_m}^\dagger\,CL_{\mathrm{pa},x_m}^\dagger\,CL_{y_m,x_m}\Bigr]Y_\gamma\,,
\ee
where the factors for different $m$ commute. For $n=1$ this reduces to the single-edge gate $CL_{\mathrm{pa},x_1}^\dagger$ conjugated by $CL_{y_1,x_1}$, which transports the error to the basepoint $v_0$.

\paragraf{Charge Transfer.} After the flux stage the data qudits are in $E_\xi^\mu\ket\psi$, with $A_{v_1}^k\ket\psi=\ket\psi$. By the first line of \eqref{eq:ribbon-endcov}, the ribbon behaves at $v_1$ as an edge entering $v_1$, so the controlled vertex operator $U_{v_1}$ of \eqref{U-def} acts as in \eqref{U-on-error-Z},
\be\label{eq:ribbon-U-on-error}
U_{v_1}\ket1_{\mathrm{va}}\ot E_\xi^\mu\ket\psi=\sqrt{d_\Gamma}\sum_{c=1}^{d_\Gamma}Z^\Gamma_{ac}(\gamma)Z^\Gamma_{cb,\mathrm{va}}\bigl(\ket1_{\mathrm{va}}\ot\ket\psi\bigr)\,.
\ee
The operator in the sum has eigenvalue $\Gamma(\hat y_nw)_{ab}$ on $\ket{y_1,\dots,y_n}\ot\ket w_{\mathrm{va}}$, so the single-edge gate $CL_{i,\mathrm{va}}$ of \eqref{CL-def} is replaced by the controlled multiplication by the ordered product along $\gamma$,
\be\ba\label{eq:ribbon-CL}
CL_{\gamma,\mathrm{va}}&:=CL_{y_1,\mathrm{va}}CL_{y_2,\mathrm{va}}\cdots CL_{y_n,\mathrm{va}}\,,\cr
CL_{\gamma,\mathrm{va}}\ket{\vec y}\ot\ket w_{\mathrm{va}}&=\ket{\vec y}\ot\ket{\hat y_nw}_{\mathrm{va}}\,,
\ea\ee
with $\vec y:=(y_1,\dots,y_n)$ and the rightmost gate applied first. As in \eqref{CL-action} and \eqref{CL-vacuum}, conjugation by $CL_{\gamma,\mathrm{va}}$ maps the operator in \eqref{eq:ribbon-U-on-error} to $Z^\Gamma_{ab,\mathrm{va}}$ and $CL_{\gamma,\mathrm{va}}$ fixes $\ket1_{\mathrm{va}}$, so
\be\label{eq:ribbon-CL-U-result}
CL_{\gamma,\mathrm{va}}\,U_{v_1}\ket1_{\mathrm{va}}\ot E_\xi^\mu\ket\psi=\ket{\Gamma_{ab}}_{\mathrm{va}}\ot\ket\psi\,.
\ee
The data qudits return to $\ket\psi$ and the Fourier transform \eqref{eq:fourier-gate} followed by a measurement in the group basis returns $\mu=(\Gamma,a,b)$. The complete recovery of a ribbon error therefore uses two ancillas, one per label, and consists of controlled group multiplications and one Fourier transform, exactly as for a single edge. For $n=1$ with $\gamma$ a single edge the ribbon transfers reduce to the single-edge transfers of Secs.~\ref{sec:charge-transfer} and \ref{sec:flux-transfer}, up to the transport by $y_1$ in \eqref{eq:ribbon-CF-circuit}.

\paragraf{Junctions and Boundary Ends.} Several ribbons with pairwise disjoint supports that share one end site and whose other end sites carry distinct, untruncated checks form a multi-ribbon support, with one pair of labels per ribbon. 
Each nontrivial label pair is detected at the private far end by Lemma~\ref{lemma:ribbon-detectable} {if the shared site is $s_1$. If the shared site is $s_0$, the flux is detected at $p_1$ by the last line of \eqref{eq:ribbon-ends}, $B_{p_1}F_\xi^{h,g}B_{p_1}=F_\xi^{h,g}B_{p_1}^{g^{-1}h^{-1}g}B_{p_1}=0$ for $h\neq\id$, and the charge is detected at $v_0$ by averaging the second line of \eqref{eq:ribbon-endcov} over $k$, so Lemma~\ref{lemma:ribbon-detectable} holds with $(p_0,v_1)$ replaced by $(p_1,v_0)$}, so the proof of Theorem~\ref{thm:KLribbon} applies with \eqref{eq:ribbon-orth} tensored over the ribbons, and the transfers above are applied ribbon by ribbon, starting from the far ends as in Sec.~\ref{sec:gauge-fixed-correction}. 
A ribbon may also end on a boundary without leaving a boundary excitation when the anyon it deposits there condenses \cite{Beigi:2010htr}. At a smooth end $v_1$, $B_{p_1}=\bbI$ while the flux is still read at $p_0$, so Lemma~\ref{lemma:ribbon-detectable} and the flux transfer are unchanged. At a rough end $v_1$, $A_{v_1}=\bbI$ and the charge is read at $v_0$ instead, where by the second line of \eqref{eq:ribbon-endcov} the ribbon behaves as an edge leaving $v_0$, so the charge transfer is the second line of \eqref{charge-transfer-circuit}, with $CL_{\gamma,\mathrm{va}}^\dagger$ in place of $CL_{i,\mathrm{va}}^\dagger$.


\twocolumngrid
\bibliographystyle{ytphys}
\small 
\baselineskip=.94\baselineskip
\let\bbb\bibitem\def\bibitem{\itemsep4pt\bbb}
\bibliography{ref}
\end{document}